\documentclass[aip,jmp,reprint]{revtex4-1}

\usepackage{graphicx}% Include figure files
\usepackage{dcolumn}% Align table columns on decimal point
\usepackage{booktabs}
\usepackage{multirow}
\usepackage{pdfpages}
\usepackage{amsthm,amsmath,amssymb}
\newcolumntype{C}{>{$}c<{$}}
\AtBeginDocument{
\heavyrulewidth=.08em
\lightrulewidth=.05em
\cmidrulewidth=.03em
\belowrulesep=.65ex
\belowbottomsep=0pt
\aboverulesep=.4ex
\abovetopsep=0pt
\cmidrulesep=\doublerulesep
\cmidrulekern=.5em
\defaultaddspace=.5em
}
\newcommand{\RN}[1]{%
  \textup{\uppercase\expandafter{\romannumeral#1}}%
}
\theoremstyle{definition}
\newtheorem{definition}{Definition}%[section]
\newtheorem{theorem}{Theorem}
\newtheorem{lemma}[theorem]{Lemma}

\DeclareMathOperator\erf{erf}
\DeclareMathOperator\erfc{erfc}

\begin{document}

\preprint{AIP/123-QED}

\title{Fast evaluation of solid harmonic Gaussian integrals for local resolution-\\of-the-identity methods and range-separated hybrid functionals} %Title of paper

% repeat the \author .. \affiliation  etc. as needed
% \email, \thanks, \homepage, \altaffiliation all apply to the current author.
% Explanatory text should go in the []'s, 
% actual e-mail address or url should go in the {}'s for \email and \homepage.
% Please use the appropriate macro for the type of information

% \affiliation command applies to all authors since the last \affiliation command. 
% The \affiliation command should follow the other information.

\author{Dorothea Golze}
\email[]{dorothea.golze@chem.uzh.ch}
%\homepage[]{Your web page}
%\thanks{}
%\altaffiliation{}
\affiliation{Department of Chemistry, University of Z\"urich, Winterthurerstrasse 190, CH-8057 Z\"urich, Switzerland}
\author{Niels Benedikter}
\affiliation{QMath, Department of Mathematical Sciences, University of Copenhagen, Universitetsparken 5, 2100 K{\o}benhavn, Denmark}
\author{Marcella Iannuzzi}
\affiliation{Department of Chemistry, University of Z\"urich, Winterthurerstrasse 190, CH-8057 Z\"urich, Switzerland}
\author{Jan Wilhelm}
\affiliation{Department of Chemistry, University of Z\"urich, Winterthurerstrasse 190, CH-8057 Z\"urich, Switzerland}
\author{J\"urg Hutter}
\affiliation{Department of Chemistry, University of Z\"urich, Winterthurerstrasse 190, CH-8057 Z\"urich, Switzerland}
% Collaboration name, if desired (requires use of superscriptaddress option in \documentclass). 
% \noaffiliation is required (may also be used with the \author command).
%\collaboration{}
%\noaffiliation

\date{\today}

\begin{abstract}
An integral scheme for the efficient evaluation of two-center integrals over contracted solid harmonic Gaussian functions is presented. Integral expressions 
are derived for local operators that depend on the position vector of one of the two Gaussian centers. These expressions are then used to derive the 
formula for three-index overlap integrals where two of the three Gaussians are located at the same center. The efficient evaluation of the latter is essential for 
local resolution-of-the-identity techniques that employ an overlap metric. We compare the performance of our integral scheme to the widely used Cartesian 
Gaussian-based method of Obara and Saika (OS). Non-local interaction potentials such as standard Coulomb, modified Coulomb and Gaussian-type operators, that 
occur in range-separated hybrid functionals, are also included in the performance tests. The speed-up with respect to the OS scheme is up to three orders of 
magnitude for both, integrals and their derivatives. In particular, our method is increasingly efficient for large angular momenta and highly contracted basis 
sets. 
\end{abstract}

% insert suggested PACS numbers in braces on next line

\maketitle %\maketitle must follow title, authors, abstract and \pacs

\section{Introduction}
The rapid analytic evaluation of two-center Gaussian integrals is important for many molecular simulation methods. For example, 
Gaussian functions are widely used as orbital basis in quantum mechanical (QM) calculations and are implemented in many electronic-structure 
codes.\cite{turbomole,Hutter2014,dalton,molpro,gamess,g09} Gaussians are further used at lower level of theory to model charge distributions in  molecular 
mechanics\cite{Piquemal2006,Cisneros2006,Elking2007,Gresh2007,Elking2010,Cisneros2012,Simmonett2014,Chaudret2014,Giese2015} (MM), 
semi-empirical\cite{Koskinen2009,Bernstein2002,Giese2005} and hybrid QM/MM methods.\cite{Giese2007,Golze2013,Giese2016} The Gaussian-based treatment of the 
electrostatic interactions requires the evaluation of two-center Coulomb integrals. \par
The efficient evaluation of two-center integrals is also important at the Kohn-Sham density functional theory (KS-DFT) level, in particular for hybrid density 
functionals. In order to speed-up the evaluation of the Hartree-Fock exchange term, the exact evaluation of the four-center integrals can be replaced by 
resolution-of-the-identity (RI) approximations.\cite{Sodt2008,Manzer2015,Ihrig2015,Levchenko2015} Especially, when an overlap metric is employed, the efficient 
evaluation of two-center integrals is required. The interaction potential can take different functional forms dependent on the hybrid 
functionals.\cite{Guidon2008} The most popular potential is the standard Coulomb operator employed in well-established functionals such as 
PBE0\cite{Perdew1996,Ernzerhof1997,Ernzerhof1999} and B3LYP.\cite{Becke1993,Lee1988,Vosko1980} A short-range Coulomb potential is, e.g., employed for the HSE06 
functional,\cite{Heyd2003,Heyd2006,Krukau2006}
whereas 
a combination of long-range Coulomb and Gaussian-type potential is used for the MCY3 functional.\cite{Cohen2007}\par
Gaussian overlap integrals, in the following denoted by $(ab)$, are computed in semi-empirical methods\cite{Koskinen2009} and QM approaches such as 
Hartree-Fock and KS-DFT. The efficient computation of $(ab)$ is not of major importance for QM methods since their contribution to the total computational cost 
is negligible. However, the efficient evaluation of the three-index overlap integrals $(ab\tilde{a})$, where two functions are located at the same center, is 
essential for local RI approaches that use an overlap metric.\cite{Baerends1973,FonsecaGuerra1998,Velde2001} Employing local RI in KS-DFT, the atomic pair 
densities are approximated by an expansion in atom-centered auxiliary functions. In order to solve the RI equations, it is necessary to calculate $(ab\tilde{a})$ for 
each pair where $a,b$ refers to orbital functions at atoms A and B and $\tilde{a}$ to the auxiliary function at A. The evaluation of $(ab\tilde{a})$ is 
computationally expensive because the auxiliary basis set is 3-5 times larger than the orbital basis set. A rapid evaluation of $(ab\tilde{a})$ is important to 
ensure that the computational overhead of the integral calculation is not larger than the speed-up gained by the RI .\par
Two-center integrals with the local operator $r_a^{2n}$ ($n\in\mathbb{N}$), where $r_a$ depends on the center of one of the Gaussian functions, are required 
for special projection and expansion techniques. For example, these integrals are used for projection of the primary orbital basis on smaller, adaptive basis 
sets.\cite{Schuett2016}\par
Numerous schemes for the evaluation of Gaussian integrals have been 
proposed based on Cartesian Gaussian,\cite{Dupuis1976,Obara1986,HeadGordon1988,Lindh1991,Bracken1997,Gill2000,Ahlrichs2006} Hermite Gaussian 
\cite{McMurchie1978,Helgaker1992,Klopper1992,Reine2007} and  solid or spherical harmonic Gaussian 
functions.\cite{Dunlap1990,Dunlap2001,Dunlap2003,Dunlap2013,Reine2007,Giese2008,Kuang1997,Kuang1997a} For a review of Gaussian integral schemes see  
Ref.~\onlinecite{Reine2012}. A very popular approach is the Obara-Saika (OS) 
scheme,\cite{Obara1986} which employs a recursive formalism over primitive Cartesian Gaussian functions. However, electronic-structure codes utilize spherical 
harmonic Gaussians (SpHGs) since the number of SpHGs is equal or smaller than the number of Cartesian Gaussians, i.e. for fixed angular momentum $l$, $(2l+1)$ 
SpHGs compare to $(l+1)(l+2)/2$ Cartesian Gaussians. Furthermore, Gaussian basis sets are often constituted of contracted functions. Thus, the primitive 
Cartesian integrals obtained from the OS recursion are subsequently contracted and transformed to SpHGs.\par
In this work, we further develop an alternative integral scheme\cite{Dunlap1990,Dunlap2001,Dunlap2003,Giese2008} that employs contracted solid harmonic 
Gaussians (SHGs). The latter are closely related to SpHG functions  and differ solely by a constant factor. The SHG integral scheme is based on the application of the spherical tensor gradient operator (STGO).\cite{Weniger1985,Weniger2005} The expressions resulting from Hobson's theorem of 
differentiation\cite{Hobson1892} contain an angular momentum term that is independent on the exponents and contraction coefficients. This term is obtained by relatively simple recursions. It can be pre-computed and re-used multiple times for all functions in the basis set with the same $l$ and $m$ quantum number. The integral and derivative evaluation requires the contraction of a set of auxiliary integrals over $s$ functions and their scalar derivatives. The same contracted quantity is re-used several times for the evaluation of functions with the same set of exponents and contraction coefficients, but different angular dependency $m$. Unlike for Cartesian functions, subsequent transformation and contraction steps are not required.\par
This work is based on Refs.~\onlinecite{Giese2008,Giese2016_book}, where the two-index integral expressions for the overlap operator and general non-local 
operators are given. We extend the SHG scheme to the local operator $r_a^{2n}$ and derive formulas for the integrals $(a|r_a^{2n}|b)$. The latter are 
fundamental for the subsequent derivation of the three-index overlap integral $(ab\tilde{a})$. The performance of the SHG method is compared to the OS scheme. 
We also include integrals with different non-local operators such as standard Coulomb, modified Coulomb or Gaussian-type operators in our comparison. \par
In the next section, the expressions for the integrals and their Cartesian derivatives are given followed by details on the implementation of the integral 
schemes. The performance of the SHG scheme is then discussed in terms of number of operations and empirical timings. The derivations of the expressions for $(a|r_a^{2n}|b)$ are given in Appendix~\ref{app:proof_product} and \ref{app:proof_sra2m}.
%---------------- Integral and derivative evaluation
%
\section{Integral and derivative evaluation}
After introducing the relevant definitions and notations, we summarize the work of Giese and York\cite{Giese2008} in 
Section~\ref{sec:integrals_r12}. The integral expressions of $(a|r_a^{2n}|b)$ and $(ab\tilde{a})$ are then derived in Sections~\ref{sec:integrals_ra2n} and 
\ref{sec:integrals_aba}, respectively. Subsequently, the formulas for the Cartesian derivatives are given (Section~\ref{sec:cartesian_derivatives}) as well as 
the details on the computation of the angular-dependent term in the SHG integral expressions (Section~\ref{sec:real_translation_matrix}).
%--------------------------Definitions
\subsection{Definitions and notations}
\label{sec:shg_definitions_and_notations}
The notations used herein correspond to Refs.~\onlinecite{Watson2004,Helgaker_book,Giese2008} unless otherwise indicated. 
An unnormalized, primitive SHG function is defined as
\begin{equation}
 \chi_{l,m}(\alpha,\mathbf{r}) = C_{l,m}(\mathbf{r})\exp{(-\alpha r^2)}
 \label{eq:SHG_primitive}
\end{equation}
where the complex solid harmonics $C_{l,m}(\mathbf{r})$,
\begin{equation}
 C_{l,m}(\mathbf{r}) = \sqrt{\frac{4\pi}{2l+1}}r^lY_{l,m}(\theta,\phi),
 \label{eq:solid_harominc_C}
\end{equation}
are obtained by rescaling the spherical harmonics $Y_{l,m}(\theta,\phi)$.  Contracted SpHG functions 
$\varphi_{l,m}(\mathbf{r})$ are constructed as linear combination of the primitive SHG functions
\begin{equation}
 \varphi_{l,m}(\mathbf{r}) = N_l\sum_{\alpha\in A}c_\alpha\chi_{l,m}(\alpha,\mathbf{r}),
 \label{eq:contracted_phi}
\end{equation}
where  $\{c_{\alpha}\}$ are the contraction coefficients for the set of exponents $A=\{\alpha\}$ and $N_l$ is the normalization constant given 
by\cite{Schlegel1995}
\begin{equation}
 N_{l} =K_l\left[\sum_{\alpha \in A}\sum_{\hat{\alpha}\in A} \frac{\pi^{1/2}(2l+2)!c_{\alpha}c_{\hat{\alpha}}}{2^{2l+3}(l+1)!(\alpha + 
\hat{\alpha})^{l+3/2}}\right]^{-1/2}.
\end{equation}
The factor
\begin{equation}
  K_l = \sqrt{\frac{2l+1}{4\pi}}
  \label{eq:Kl}
\end{equation}
is included in the normalization constant to convert from SHG to SpHG functions.\par
In the following, the absolute value of the $m$ quantum number is denoted by
\begin{equation}
 \mu = |m|.
 \label{eq:mu}
\end{equation}
Furthermore, we use the notations,
\begin{equation}
 \mathbf{r}_a=\mathbf{r}-\mathbf{R}_a,\quad
\mathbf{r}_b=\mathbf{r}-\mathbf{R}_b,\quad
R_{ab}^2 = |\mathbf{R}_a-\mathbf{R}_b|^2,
\end{equation}
where $\mathbf{R}_a$ references the position of the Gaussian center A and $\mathbf{R}_b$ the position of center B. The scalar derivative of $X(r^2)$ with 
respect to $r^2$ is denoted by
\begin{equation}
 X^{(k)}(r^2) =\left(\frac{\partial}{\partial r^2}\right)^{k}X(r^2). 
 \label{eq:deriv}
\end{equation}
%
%
%------------------------ Integrals (a|O|b)
%
\subsection{Integrals \boldmath{$(a|\mathcal{O}|b)$}}
\label{sec:integrals_r12}
In this section, the expression to compute the two-center integral $(a|\mathcal{O}|b)$ is given which is defined as
\begin{equation}
 \begin{split}
 (a|\mathcal{O}|b) = 
\iint&\varphi_{l_a,m_a}(\mathbf{r}_1-\mathbf{R}_a)\mathcal{O}(\mathbf{r}_1-\mathbf{r}_2)\\&\times\varphi_{l_b,m_b}(\mathbf{r}_2-\mathbf{R}_b) 
d\mathbf{r}_1d\mathbf{r}_2.
 \end{split}
\label{eq:contractedr12int}
\end{equation}
$\varphi_{l_a,m_a}(\mathbf{r}_a)$ and $\varphi_{l_b,m_b}(\mathbf{r}_b)$ are contracted SpHG functions as defined in Equation~\eqref{eq:contracted_phi}, which are 
centered at  $\mathbf{R}_a$ and $\mathbf{R}_b$, respectively. $\mathcal{O}(\mathbf{r})$ is an operator that is explicitly independent on the position vectors 
$\mathbf{R}_a$ or $\mathbf{R}_b$. Such operators are, e.g., the non-local Coulomb operator $\mathcal{O}(\mathbf{r})=1/r$ or the local overlap  
$\mathcal{O}(\mathbf{r})=\delta(\mathbf{r})$. \par
The derivation for an efficient expression to compute $(a|\mathcal{O}|b)$ follows Ref.~\onlinecite{Giese2008}. It is based on Hobson's theorem\cite{Hobson1892} 
of differentiation, which states that
\begin{equation}
 C_{l,m}(\nabla)f(r^2) = 2^lC_{l,m}(\mathbf{r})\left(\frac{\partial}{\partial r^2}\right)^l f(r^2),
 \label{eq:Hobson}
\end{equation}
where the differential operator $C_{l,m}(\nabla)$ is called STGO. The differential operator is obtained by replacing $\mathbf{r}$ in the solid harmonic 
$C_{l,m}(\mathbf{r})$ by $\nabla=(\partial /\partial x, \partial /\partial y, \partial /\partial z)$. The derivation of the $(a|\mathcal{O}|b)$ integrals 
starts by noting that $\exp(-\alpha r^2)$ is an eigenfunction of $(\displaystyle{\partial /\partial r^2})^l$ with the eigenvalue $(-\alpha)^l$. Using Equation~\eqref{eq:Hobson} and the definition of primitive SHGs from Equation~\eqref{eq:SHG_primitive}, the primitive SHG at center $\mathbf{R}_a$ can be rewritten as
\begin{equation}
  \chi_{l,m}(\alpha,\mathbf{r}_a) = \frac{C_{l,m}(\nabla_a)\exp{\left(-\alpha r_a^2\right)}}{(2\alpha)^l},
  \label{eq:chi_STGO}
\end{equation} 
where $C_{l,m}(\nabla_a)$ acts on $\mathbf{R}_a$. Inserting Equation~\eqref{eq:SHG_primitive} for $s$ functions,  
$\chi_{0,0}(\alpha,\mathbf{r})=\exp{\left(-\alpha r^2\right)}$, yields
\begin{equation}
  \chi_{l,m}(\alpha,\mathbf{r}_a) = \frac{C_{l,m}(\nabla_a)\chi_{0,0}(\alpha,\mathbf{r}_a)}{(2\alpha)^l}.
  \label{eq:primitive_shg_STGO}
\end{equation}
Inserting the STGO formulation of $\chi_{l,m}$ from Equation~\eqref{eq:primitive_shg_STGO} in Equation~\eqref{eq:contractedr12int} gives 
\begin{equation}
  (a|\mathcal{O}|b) = C_{l_a,m_a}(\nabla_a) C_{l_b,m_b}(\nabla_b) O_{l_a,l_b}(R_{ab}^2).
 \label{eq:contracted_r12int_ab_STGO}
\end{equation}
The contracted integral over $s$ functions is denoted by
\begin{equation}
 O_{l_a,l_b}(R_{ab}^2)  =  N_{l_a} N_{l_b}
\sum_{\alpha \in  A}\sum_{\beta \in  B}
\frac{c_{\alpha}c_{\beta} }{(2\alpha)^{l_a}(2\beta)^{l_b}}
(0_a|\mathcal{O}|0_b),
 \label{eq:O_rab2}
\end{equation}
where $c_{\alpha}$ and $c_{\beta}$ are the expansion coefficients of $\varphi_{l_a,m_a}(\mathbf{r}_a)$ and $\varphi_{l_b,m_b}(\mathbf{r}_b)$, respectively, with corresponding exponents~$\alpha$ and $\beta$. The integral $(0_a|\mathcal{O}|0_b)$ over primitive 
$s$ functions is given by
\begin{equation}
 \begin{split}
 (0_a|\mathcal{O}|0_b) 
=\iint&\chi_{0,0}(\alpha,\mathbf{r}_1-\mathbf{R}_a)\mathcal{O}(\mathbf{r}_1-\mathbf{r}_2)\\&\times\chi_{0,0}(\beta,\mathbf{r}_2-\mathbf{R}_b)d\mathbf{r}
_1d\mathbf{r}_2.
 \end{split}
 \label{eq:ss_r12int}
\end{equation}
The analytic expressions of $(0_a|\mathcal{O}|0_b)$ for the overlap and different non-local operators are given in Table~S1, see Supplementary Information (SI).
 Application of the product and differentiation rules for the 
STGO\cite{Hu2000,Dunlap1990,Dunlap2001,Dunlap2003} finally yields 
\begin{equation}
 \begin{split}
  (a|\mathcal{O}|b) = &(-1)^{l_b} A_{l_a,m_a}A_{l_b,m_b} 
\\[0.5em]&\times\sum_{j=0}^{\mathrm{min}(l_a,l_b)}2^{l_a+l_b-j}O^{(l_a+l_b-j)}_{l_a,l_b}(R_{ab}^2)\\&\times(2j-1)!!\sum_{\kappa=0}^jB_{j,\kappa}Q_{l_a,\mu_a,l_b
,\mu_b,j,\kappa}^{c/s,c/s}
(\mathbf{R}_{ab}),
 \end{split}
 \label{eq:final_r12int_ab}
\end{equation}
where the prefactors are 
\begin{align}
  A_{l,m} &= (-1)^{m} \sqrt{(2-\delta_{m,0}) (l+m)!(l-m)!}\label{eq:A}
\\[0.5em]
   B_{j,\kappa} &= \frac{1}{(2-\delta_{\kappa,0})(j+\kappa)!(j-\kappa)!}\,.\label{eq:B}
 \end{align}
and $n!!$ denotes the double factorial. The  superscript on $O_{l_a,l_b}(R_{ab}^2)$ in Equation~\eqref{eq:final_r12int_ab} denotes the scalar derivative with 
respect to $R_{ab}^2$, see Equation~\eqref{eq:deriv} and~\eqref{eq:O_rab2},
\begin{equation}
 \begin{split}
 O^{(k)}_{l_a,l_b}(R_{ab}^2) = &  N_{l_a} N_{l_b}
\sum_{\alpha \in  A}\sum_{\beta \in  B}
\frac{c_{\alpha}c_{\beta} }{(2\alpha)^{l_a}(2\beta)^{l_b}}
\\&
\times\left(\frac{\partial}{\partial R_{ab}^2}\right)^{k}(0_a|\mathcal{O}|0_b).
 \end{split}
 \label{eq:O_rab2_deriv}
\end{equation}
Since $s$ functions contain no angular dependency, $(0_a|\mathcal{O}|0_b)$ is a function of $\mathbf{R}_{ab}$ (or equivalently, $R_{ab}^2$), see Table~S1 (SI). 
Therefore, the derivative in Equation~\eqref{eq:O_rab2_deriv} is well-defined.
The integral $O_{l_a,l_b}(R_{ab}^2)$ can be interpreted as the monopole result of the expansion given in Equation~\eqref{eq:final_r12int_ab}.\par
The expression given in Equation~\eqref{eq:final_r12int_ab} depends further on $Q_{l_a,\mu_a,l_b,\mu_b,j,\kappa}^{c/s,c/s}(\mathbf{R}_{ab})$, where 
$\mu=|m|$. Positive $m$ values refer to a cosine ($c$) component and negative $m$ to a sine ($s$) component, i.e.
\begin{equation}
\begin{split}
 Q_{a,b,j,\kappa}^{cc}(\mathbf{R}_{ab}): & \hspace{1em}m_a,m_b\geq0\\
 Q_{a,b,j,\kappa}^{cs}(\mathbf{R}_{ab}): & \hspace{1em}m_a\ge0, m_b<0 \\
 Q_{a,b,j,\kappa}^{sc}(\mathbf{R}_{ab}): & \hspace{1em}m_a<0, m_b\ge0\\
 Q_{a,b,j,\kappa}^{ss}(\mathbf{R}_{ab}): & \hspace{1em}m_a,m_b<0\,.
\end{split}
\label{eq:Qref}
\end{equation}
Note that we used the abbreviation $a,b$ for the indices $(l_a,\mu_a,l_b,\mu_b)$ in Equation~\eqref{eq:Qref}. 
Details on the calculation of $Q_{a,b,j,\kappa}^{c/s,c/s}(\mathbf{R}_{ab})$ can be found in Section~\ref{sec:real_translation_matrix}.\par
%
%------------------------ Integrals (a|ra^2n|b)
\subsection{Integrals \boldmath{$(a|r_a^{2n}|b)$}}
\label{sec:integrals_ra2n}
The integrals $(a|r_a^{2n}|b)$,
\begin{equation}
\label{eq:ara2nb_def}
 (a|r_a^{2n}|b) = \int\varphi_{l_a,m_a}(\mathbf{r}_a)r_a^{2n}\varphi_{l_b,m_b}(\mathbf{r}_b) d\mathbf{r}\,,
\end{equation}
 $n\in \mathbb{N}$, are fundamental for the derivation of the overlap matrix elements $(ab\tilde{a})$ with two Gaussians at center $\mathbf{R}_a$, which are 
discussed in the next section. Since the operator $r_a^{2n}$ depends on the position $\mathbf{R}_a$,  Equations~\eqref{eq:O_rab2} and 
\eqref{eq:final_r12int_ab} cannot be  adapted by replacing $(0_a|\mathcal{O}|0_b)$ with $(0_a|r_a^{2n}|0_b)$.
Consequently, new expressions for computing~$(a|r_a^{2n}|b)$ are derived in this section.\par
Since $r_a^{2n}$ depends on $\mathbf{R}_a$, $C_{l,m}(\nabla_a)$ is acting on the product of $\chi_{l,m}(\alpha,\mathbf{r}_a)$ and $r_a^{2n}$, 
\begin{equation}
 \chi_{l,m}(\alpha,\mathbf{r}_a)r_a^{2n} = C_{l,m}(\mathbf{r}_a) \exp{\left(-\alpha r_a^2\right)} r_a^{2n}.
\end{equation}
The expression of this product in terms of the STGO $C_{l,m}(\nabla_a)$ is obtained using Hobson's theorem, 
\begin{equation}
 \begin{split}
   \chi_{l,m}(\alpha,\mathbf{r}_a)r_a^{2n} = \frac{C_{l,m}(\nabla_a) }{(2\alpha)^{l}}\sum_{j=0}^n&\binom{n}{j} 
\frac{(l+j-1)!}{(l-1)!\alpha^j}\\&\times\exp{\left(-\alpha r_a^2\right)} r_a^{2(n-j)}.
 \end{split}
\label{eq:generic_chira2n}
\end{equation}
The derivation of Equation~\eqref{eq:generic_chira2n}  is given in Appendix~\ref{app:proof_product}.\par
Inserting Equations~\eqref{eq:primitive_shg_STGO} and \eqref{eq:generic_chira2n} in Equation~\eqref{eq:ara2nb_def} yields
\begin{equation}
 (a|r_a^{2n}|b) = C_{l_a,m_a}(\nabla_a) C_{l_b,m_b}(\nabla_b) T_{l_a,l_b}(R_{ab}^2),
 \label{eq:ara2mb_STGO}
\end{equation}
where $T_{l_a,l_b}(R_{ab}^2)=T_{l_a,l_b}^{(0)}(R_{ab}^2)$ is again the monopole results for the integral given in Equation~\eqref{eq:ara2nb}. The derivation 
follows now the same procedure as for the integrals 
$(a|\mathcal{O}|b)$ and yields
\begin{equation}
 \begin{split}
  (a|r_a^{2n}|b) = & (-1)^{l_b} A_{l_a,m_a}A_{l_b,m_b}\\[0.5em]&\times 
\sum_{j=0}^{\mathrm{min}(l_a,l_b)}2^{l_a+l_b-j}T^{(l_a+l_b-j)}_{l_a,l_b}(R_{ab}^2)\\&\times(2j-1)!!\sum_{\kappa=0}^jB_{j,\kappa}Q_{l_a,\mu_a,l_b,\mu_b,j,\kappa}
^{c/s,c/s}(\mathbf{R
}_{ab}),
 \end{split}
 \label{eq:ara2nb}
\end{equation}
where the scalar derivative of $T_{l_a,l_b}(R_{ab}^2)$ with respect to $R_{ab}^2$ is 
\begin{equation}
\begin{split}
 T^{(k)}_{l_a,l_b}(R_{ab}^2) = & N_{l_a} N_{l_b}\sum_{\alpha \in  A}\sum_{\beta \in  B}\frac{c_{\alpha}c_{\beta} 
}{(2\alpha)^{l_a}(2\beta)^{l_b}}\\&\times\sum_{j=0}^n\binom{n}{j} 
\frac{(l_a+j-1)!}{(l_a-1)!\alpha^j}(0_a|r_a^{2(n-j)}|0_b)^{(k)}.
\end{split}
 \label{eq:T_rab2_deriv}
\end{equation}
The integral over primitive $s$ functions is
\begin{align}
 (0_a|r_a^{2m}|0_b) &= \int\chi_{0,0}(\alpha,\mathbf{r}_a)r_a^{2m}\chi_{0,0}(\beta,\mathbf{r}_b)d\mathbf{r}\label{eq:sra2ms_def}\\
                    &= \frac{\pi^{3/2}\exp(-\rho R_{ab}^2)}{2^mc^{m+3/2}}\sum_{j=0}^mI_j^{\alpha,\beta,m}(R_{ab}^2) \label{eq:sra2ms}
\end{align}
with $c=\alpha+\beta$ and $\rho =\alpha\beta/c$ and
\begin{equation}
 I_j^{\alpha,\beta,m}(R_{ab}^2) = 2^j\frac{(2m+1)!! \binom{m}{j}}{(2j+1)!!} \frac{\beta^{2j}}{c^j}R_{ab}^{2j}.
 \label{eq:I_jab}
\end{equation}%
The proof of Equation~\eqref{eq:sra2ms} is similarly elaborate as for Equation~\eqref{eq:generic_chira2n} and is given in Appendix~\ref{app:proof_sra2m}.  The 
derivatives of $(0_a|r_a^{2m}|0_b)$ are obtained by applying the Leibniz rule of differentiation to Equation~\eqref{eq:sra2ms}
\begin{equation}
 \begin{split}
   (0_a|&r_a^{2m}|0_b)^{(k)} \\=& 
\frac{\pi^{3/2}\exp(-\rho R_{ab}^2)}{2^mc^{m+3/2}}\sum_{i=0}^{\mathrm{min}(m,k)}\binom{k}{i}\left(-\rho \right)^{k-i}\\&\times\sum_{j=i}
^m\left(\frac{\partial}{\partial R_{ab}^2}\right)^{i}I_j^{\alpha,\beta,m}(R_{ab}^2).
 \end{split}
 \label{eq:devsra2ms}
\end{equation}
%-------------------------(a|b|a) Integrals
\subsection{Overlap integrals \boldmath{$(ab\tilde{a})$}} 
\label{sec:integrals_aba}
The three-index overlap integral $(ab\tilde{a})$ includes two functions at center $\mathbf{R}_a$ and is defined by 
\begin{equation}
 (ab\tilde{a})
 =\int\varphi_{l_a,m_a}(\mathbf{r}_a)\varphi_{\tilde{l}_a,\tilde{m}_a}(\mathbf{r}_a)\varphi_{l_b,m_b}(\mathbf{r}_b) d\mathbf{r}.
\label{eq:aba_integrals}
\end{equation}
In traditional Cartesian Gaussian-based schemes, the product of the two Cartesian functions at center $\mathbf{R}_a$ is obtained by adding 
exponents and angular momenta of both Gaussians, respectively. The result is a new Cartesian Gaussian at $\mathbf{R}_a$. The integral evaluation  
proceeds then as for the two-index overlap integrals $(ab)$. In the SHG scheme on the other hand, the product of two SHG functions at the same center is 
obtained by a Clebsch-Gordan (CG) expansion of the spherical harmonics. In the following, the expression of this expansion in terms 
of the STGO is derived and used to obtain the integral formula.\par
Employing the definitions given in Equations~\eqref{eq:SHG_primitive} and \eqref{eq:solid_harominc_C}, 
the product of two primitive SHG functions at $\mathbf{R}_a$ can be written as
\begin{align}
  \nonumber
  \chi_{l,m}&(\alpha, \mathbf{r}_a) \chi_{\tilde{l},\tilde{m}}(\tilde{\alpha}, \mathbf{r}_a)\\&= 
  C_{l,m}(\mathbf{r}_a) \exp(-\alpha r_a^2)C_{\tilde{l},\tilde{m}}(\mathbf{r}_a) \exp(-\tilde{\alpha}r_a^2) \\[5pt]
  &=\lambda\exp\left(-\alpha^{\prime} r_a^2\right)  Y_{l,m}(\theta,\phi) Y_{\tilde{l},\tilde{m}}(\theta,\phi) r_a^{l+\tilde{l}} \label{eq:product_SHG}
\end{align}
where $\alpha^{\prime}=\alpha+\tilde{\alpha}$ and 
\begin{equation}
 \lambda=\frac{4\pi}{\sqrt{(2l+1)(2\tilde{l}+1)}}.
\end{equation}
The product of two spherical harmonics can be expanded in terms of spherical harmonics, 
\begin{equation}
  Y_{l,m}(\theta,\phi) Y_{\tilde{l},\tilde{m}}(\theta,\phi) = \sum_{L,M}G^{L,l,\tilde{l}}_{M,m,\tilde{m}} Y_{L,M}(\theta,\phi),
  \label{eq:CG_expansion}
\end{equation}
where $|l-\tilde{l}|\leq L \leq l+\tilde{l}$. $G^{L,l,\tilde{l}}_{M,m,\tilde{m}}$ are the Gaunt coefficients\cite{Gaunt1929} which are proportional to a 
product of CG coefficients.\cite{Xu1996} The expansion given in Equation~\eqref{eq:CG_expansion} is valid since 
the spherical harmonics form a complete set of orthonormal functions. A similar expansion for solid harmonics $C_{l,m}(\mathbf{r})$ is not possible because the 
latter are no basis of ${L}^2(\mathbb{R}^3)$. Inserting the CG expansion into Equation~\eqref{eq:product_SHG} [1], re-introducing solid harmonics [2] as 
defined in Equation~\eqref{eq:solid_harominc_C} and employing the definition given in Equation~\eqref{eq:SHG_primitive} [3] yields
\begin{align}
\chi_{l,m}&(\alpha, \mathbf{r}_a) \chi_{\tilde{l},\tilde{m}}(\tilde{\alpha}, \mathbf{r}_a)\nonumber\\&\stackrel{[1]}{=}   \lambda\exp\left(-\alpha^{\prime} 
r_a^2\right) 
\sum_{L,M}G^{L,l,\tilde{l}}_{M,m,\tilde{m}} Y_{L,M}(\theta,\phi) r_a^{l+\tilde{l}}\\[5pt]
 &  \stackrel{[2]}{=} \lambda \sum_{L,M}G^{L,l,\tilde{l}}_{M,m,\tilde{m}} K_L C_{L,M}(\mathbf{r}_a)\nonumber\\&\qquad\qquad\times 
\exp\left(-\alpha^{\prime} 
r_a^2\right)r_a^{l+\tilde{l}-L}\\[5pt]
   &\stackrel{[3]}{=}  \lambda\sum_{L,M}G^{L,l,\tilde{l}}_{M,m,\tilde{m}} K_L\chi_{L,M}(\alpha^{\prime},\mathbf{r}_a)r_a^{l+\tilde{l}-L},
\end{align}
where $K_L$ is defined in Equation~\eqref{eq:Kl}.
The $L$ quantum numbers of the non-vanishing contributions in the CG expansion proceed in steps of two starting from $L_{\mathrm{min}}=|l-\tilde{l}|$ to
$L_{\mathrm{max}}=l+\tilde{l}$. Thus, $l+\tilde{l}-L$ is even and we can express 
$\chi_{L,M}(\alpha^{\prime},\mathbf{r}_a)r_a^{l+\tilde{l}-L}$ in terms of the STGO using Equation~\eqref{eq:generic_chira2n},
\begin{equation}
 \begin{split}
 \chi_{l,m}&(\alpha, \mathbf{r}_a) \chi_{\tilde{l},\tilde{m}}(\tilde{\alpha}, \mathbf{r}_a)\\={}& \lambda 
\sum_{L,M}G^{L,l,\tilde{l}}_{M,m,\tilde{m}}K_L \frac{C_{L,M}(\nabla_a) }{(2\alpha^{\prime})^{L}}\\&\times\sum_{j=0}^p\binom{p}{j} 
\frac{(L+j-1)!}{(L-1)!(\alpha^{\prime})^j}\exp\left(-\alpha^{\prime} r_a^2\right) r_a^{2(p-j)}
 \end{split}
\label{eq:product_SHG_STGO}
\end{equation}
with $p=(l+\tilde{l}-L)/2$.
\par
The derivation of the integral expression for $(ab\tilde{a})$ is analogous to the $(a|\mathcal{O}|b)$ integrals. 
Inserting the STGO formulations given in Equation~\eqref{eq:chi_STGO} and Equation~\eqref{eq:product_SHG_STGO} into 
Equation~\eqref{eq:aba_integrals} yields 
\begin{equation}
 \begin{split}
 (ab\tilde{a})
=\sum_{L_a,M_a}&G^{L_a,l_a,\tilde{l}_a}_{M_a,m_a,\tilde{m}_a} C_{L_a,M_a}(\nabla_a)\\&\times C_{l_b,m_b}(\nabla_b) P_{L_a,l_a,\tilde{l}_{a},l_b}(R_{ab}^2)
 \end{split}
\end{equation}
with 
\begin{equation}
  \begin{split}
P_{L_a,l_a,\tilde{l}_{a},l_b}(&R_{ab}^2)\\ ={}&  \lambda K_{L_a} N_{l_a} N_{l_b}N_{\tilde{l}_a}\sum_{\alpha \in  A}\sum_{\beta \in 
 B}\sum_{\tilde{\alpha} \in 
 \tilde{A}}\frac{c_{\alpha}c_{\beta}c_{\tilde{\alpha}}
}{(2\alpha^{\prime})^{L_a}(2\beta)^{l_b}}\\[5pt]&\times\sum_{j=0}^p\binom{p}{j} 
\frac{(L_a+j-1)!}{(L_a-1)!(\alpha^{\prime})^j}(0_{a^{\prime}}|r_a^{2(p-j)}|0_b),
 \label{eq:Pa_aba}
 \end{split}
\end{equation}
where the dependence of $P_{L_a,l_a,\tilde{l}_{a},l_b}$ on $R_{ab}^2$ originates from the integrals over primitive $s$ functions,
\begin{equation}
 (0_{a^{\prime}}|r_a^{2m}|0_b) = \int\chi_{0,0}(\alpha^{\prime},\mathbf{r}_a)r_a^{2m}\chi_{0,0}(\beta,\mathbf{r}_b)d\mathbf{r},
\end{equation}
see Equation~\eqref{eq:sra2ms}. The derivation proceeds as for the $(a|\mathcal{O}|b)$ and $(a|r_a^{2n}|b)$ integrals yielding the final formula,
\begin{equation}
 \begin{split}
(ab\tilde{a})={}&(-1)^{l_b}A_{l_b,m_b}\sum_{L_a,M_a}G^{L_a,l_a,\tilde{l}_a}_{M_a,m_a,\tilde{m}_a}A_{L_a,M_a}\\&\times\sum_{j=0}^{\mathrm{min}(L_a,l_b)}2^{L_a+l_b-j}P^{
(L_a+l_b-j)}_{L_a,l_a,\tilde{l}_{a},l_b}(
R_{ab}^2)\\&\times(2j-1)!!\sum_{\kappa=0}^jB_{j,\kappa}Q_{L_a,|M_a|,l_b,\mu_b,j,\kappa}^{c/s,c/s}(\mathbf{R}_{ab}),
  \end{split}
\end{equation}
where the coefficients $A_{l,m}$ and $B_{j,\kappa}$ are given in Equations~\eqref{eq:A} and~\eqref{eq:B}. See Section~\ref{sec:real_translation_matrix} for the 
expressions of $Q_{L_a,|M_a|,l_b,\mu_b,j,\kappa}^{c/s,c/s}$. The superscript $(L_a{-}l_b{-}j)$ on $P_{L_a,l_a,\tilde{l}_{a},l_b}$ indicates the derivative as 
defined in Equation~\eqref{eq:deriv}.\par
The integral $(ab\tilde{a})$ can be considered as a sum of $(a|r_a^{2n}|b)$  integrals, introducing some modifications due to normalization and contraction. 
%--------------------------Derivatives
\subsection{Cartesian Derivatives}
\label{sec:cartesian_derivatives}
Cartesian derivatives are required for evaluating forces and stress in molecular simulations. The Cartesian derivatives of the integrals $(a|\mathcal{O}|b)$, $(a|r_a^{2n}|b)$ and $(ab\tilde{a})$ are obtained by applying the product 
rule to the  $R_{ab}^2$-dependent contracted quantities [Equations~\eqref{eq:O_rab2_deriv},\eqref{eq:T_rab2_deriv} and \eqref{eq:Pa_aba}] and the matrix elements of $Q_{l_a,\mu_a,l_b,\mu_b,j,\kappa}^{c/s,c/s}(\mathbf{R}_{ab})$ . The derivative of 
$(a|\mathcal{O}|b)$ [Equation~\eqref{eq:final_r12int_ab}] with respect to $\mathbf{R}_a$ is
\begin{equation}
\begin{split}
 \frac{\partial}{\partial R_{a,i}}&(a|\mathcal{O}|b) \\
  ={}& 2(R_{a,i}-R_{b,i})\\&\times\sum_{j=0}^{\mathrm{min}(l_a,l_b)}O^{(l_a+l_b-j+1)}_{l_a,l_b}(R_{ab}^2)
\widetilde{Q}_{l_a,\mu_a,l_b,\mu_b,j}^{c/s,c/s}(\mathbf{R}_{ab})\\[5pt]
   &+\sum_{j=0}^{\mathrm{min}(l_a,l_b)}O^{(l_a+l_b-j)}_{l_a,l_b}(R_{ab}^2) \frac{\partial 
\widetilde{Q}_{l_a,\mu_a,l_b,\mu_b,j}^{c/s,c/s}(\mathbf{R}_{ab})}{\partial R_{a,i}}
 \label{eq:dev_ab}
 \end{split}
\end{equation}
with $i=x,y,z$ and where we have introduced the notation
\begin{equation}
 \begin{split}
 \widetilde{Q}_{l_a,\mu_a,l_b,\mu_b,j}^{c/s,c/s}&(\mathbf{R}_{ab})\\ ={}& (-1)^{l_b}A_{l_a,\mu_a}A_{l_b,\mu_b} 2^{l_a+l_b-j} (2j-1)!!\\&\times 
\sum_{\kappa=0}^j 
B_{j,\kappa} 
Q_{l_a,\mu_a,l_b,\mu_b,j,\kappa}^{c/s,c/s}(\mathbf{R}_{ab}).
 \end{split}
 \label{eq:Q_tilde}
\end{equation}
The derivatives of $(a|r_a^{2n}|b)$ are obtained from Equation~\eqref{eq:dev_ab} by substituting $O_{l_a,l_b}(R_{ab}^2)$ by $T_{l_a,l_b}(R_{ab}^2)$. For 
$(ab\tilde{a})$, we 
replace $O_{l_a,l_b}(R_{ab}^2)$ by $P_{L_a,l_a,\tilde{l}_{a},l_b}(R_{ab}^2)$ considering additionally the CG expansion. The derivatives of $\widetilde{Q}_{l_a,\mu_a,l_b,\mu_b,j}^{c/s,c/s}$ are constructed from $(l-1)$ terms, which is explained in detail in Section~\ref{sec:real_translation_matrix}.

%
%-------- Real translation matrix
\subsection{Computation of \boldmath${Q}^{c/s,c/s}_{a,b,j,\kappa}$ and its derivatives}
\label{sec:real_translation_matrix}
${Q}^{c/s,c/s}_{a,b,j,\kappa}$, introduced in Equation~\eqref{eq:Qref}, are elements of the 2$\times$2 matrix $\mathbf{Q}_{a,b,j,\kappa}$, which is computed from the 
real translation matrix 
$\mathbf{W}_{l,m,j,\kappa}$\cite{Watson2004,Helgaker_book}
\begin{equation}
 \begin{split}
  \mathbf{Q}_{a,b,j,\kappa}(\mathbf{R}_{ab}) &= \begin{pmatrix}
      Q_{a,b,j,\kappa}^{cc} &  Q_{a,b,j,\kappa}^{cs} \\
      Q_{a,b,j,\kappa}^{sc} &  Q_{a,b,j,\kappa}^{ss}
  \end{pmatrix}_{(\mathbf{R}_{ab})}\\[5pt]
    &=\mathbf{W}_{l_a,\mu_a,j,\kappa}(-\mathbf{R}_{ab}) \mathbf{W}_{l_b,\mu_b,j,\kappa}^{\mathrm{T}}(-\mathbf{R}_{ab}).
 \end{split}
  \label{eq:Q_matrix}
\end{equation}
Note that we abbreviate the indices $(l_a,\mu_a,l_b,\mu_b)$ with $(a,b)$ in $\mathbf{Q}_{a,b,j,\kappa}$ as in Equation~\eqref{eq:Qref}.
The real translation matrix is a 2$\times$2 matrix with the elements
\begin{equation}
  \mathbf{W}_{l,m,j,\kappa}(\mathbf{R}_{ab}) = 
  \begin{pmatrix}
    W^{cc}_{l,m,j,\kappa}  &  W^{cs}_{l,m,j,\kappa} \\
    W^{sc}_{l,m,j,\kappa}  &  W^{ss}_{l,m,j,\kappa}
  \end{pmatrix}_{(\mathbf{R}_{ab})}.
\label{eq:W_matrix}
\end{equation}
The expressions for $\mathbf{W}_{l,m,j,\kappa}$ are given by\cite{Helgaker_book},
\begin{align}
   W^{cc}_{l,m,j,\kappa}(\mathbf{R}_{ab}) ={}& \left(\frac{1}{2}\right)^{\delta_{\kappa0}}\left[R^c_{l-j,m-\kappa}(-\mathbf{R}_{ab})\right.\nonumber\\&  \left.+
(-1)^{\kappa}R^c_{l-j,m+\kappa}(-\mathbf{R}_{ab})\right]\\[5pt]
   W^{cs}_{l,m,j,\kappa}(\mathbf{R}_{ab}) ={}& - R^s_{l-j,m-\kappa}(-\mathbf{R}_{ab}) \nonumber\\& + (-1)^{\kappa}R^s_{l-j,m+\kappa}(-\mathbf{R}_{ab})\\[5pt]
   W^{sc}_{l,m,j,\kappa}(\mathbf{R}_{ab}) ={}& \left(\frac{1}{2}\right)^{\delta_{\kappa0}}\left[R^s_{l-j,m-\kappa}(-\mathbf{R}_{ab}) \right.\nonumber\\&  \left. 
+ 
(-1)^{\kappa}R^s_{l-j,m+\kappa}(-\mathbf{R}_{ab})\right]\\[5pt]
   W^{ss}_{l,m,j,\kappa}(\mathbf{R}_{ab}) ={}& R^c_{l-j,m-\kappa}(-\mathbf{R}_{ab})\nonumber\\&  - (-1)^{\kappa}R^c_{l-j,m+\kappa}(-\mathbf{R}_{ab}
).
\end{align}
Here, we introduced the regular scaled solid harmonics $R_{l,m}(\mathbf{r})$ which are defined as
\begin{equation}
 R_{l,m}(\mathbf{r}) = \frac{1}{\sqrt{(l-m)!(l+m)!}}\,C_{l,m}(\mathbf{r}),
\end{equation}
where the definition of the complex solid harmonics $C_{l,m}(\mathbf{r})$ from Equation~\eqref{eq:solid_harominc_C} has been employed.
The regular scaled solid harmonics are also complex and can be decomposed into a real (cosine) and an imaginary (sine) part as
\begin{equation}
 R_{l,m}(\mathbf{r}) = R_{l,m}^c(\mathbf{r}) + iR_{l,m}^s(\mathbf{r}).
 \label{eq:scaled_SHG}
\end{equation}
The cosine and sine parts can be constructed by the following recursion relations\cite{Watson2004,Helgaker_book}
\begin{align}
 &R^c_{00} = 1, \quad R^s_{00} =0 \label{eq:regular_scaled_harmonics_recursion_start}\\[3pt] 
 &R^c_{l+1,l+1} = -\frac{xR^c_{ll}-yR^s_{ll}}{2l+2}\\[3pt]
 &R^s_{l+1,l+1} = -\frac{yR^c_{ll}+xR^s_{ll}}{2l+2}\\[3pt]
 &R^{c/s}_{l+1,m}=\frac{(2l+1)zR^{c/s}_{l,m}-r^2R^{c/s}_{l-1,m}}{(l+m+1)(l-m+1)},\quad  0\leq m < l
 \label{eq:regular_scaled_harmonics_recursion}
\end{align}
where $\mathbf{r}=(x,y,z)$.  The usage of $c/s$ in the last recurrence formula indicates that the relation is used for both, $R^c_{l,m}(\mathbf{r})$ and 
$R^s_{l,m}(\mathbf{r})$. The recursions are only valid for positive $m$. However, the regular scaled solid harmonics are also defined for negative indices and 
satisfy the following symmetry relations
\begin{equation}
  R^c_{l,-m} = (-1)^m R^c_{l,m}, \quad
  R^s_{l,-m} = -(-1)^m R^s_{l,m}.
\end{equation}
Note that these symmetry relations have to be employed for the evaluation of $R^{c/s}_{l-j,\mu-\kappa}$ since $\mu-\kappa$ can be also negative. Furthermore, 
only elements with $l-j\geq|\mu\pm\kappa|$ give non-zero contributions.\par
The elements of the transformation matrix $\mathbf{W}_{l,m,j,\kappa}$ are also defined for negative $m$. The matrix elements of 
$\mathbf{W}_{l,m,j,\kappa}$ obey the same symmetry relations with respect to sign changes of $m$,
\begin{equation}
  \begin{split}
   W^{cc/cs}_{l,\overline{m},j,\kappa} &= (-1)^m  W^{cc/cs}_{l,m,j,\kappa}\\[3pt]
   W^{sc/ss}_{l,\overline{m},j,\kappa} &=-(-1)^m W^{sc/ss}_{l,m,j,\kappa} 
   \label{eq:symmetry_relations_W}
 \end{split}
\end{equation}
where we have used the notation $\overline{m} = -m$. These symmetry relations are used for the derivatives of $Q_{l_a,\mu_a,l_b,\mu_bj,\kappa}^{c/s,c/s}$.\par
The derivatives of $Q_{l_a,\mu_a,l_b,\mu_b,j,\kappa}^{c/s,c/s}$ and equivalently of $\widetilde{Q}_{l_a,\mu_a,l_b,\mu_b,j}^{c/s,c/s}$ from 
Equation~\eqref{eq:Q_tilde} are obtained by employing the differentiation rules\cite{PerezJorda1996} of the solid 
harmonics $C_{l,m}(\mathbf{r})$. The derivatives of $C_{l,m}(\mathbf{r})$ are a linear combination of $(l-1)$ solid harmonics. Therefore, the 
gradients of $\widetilde{Q}^{c/s,c/s}_{l_a,\mu_a,l_b,\mu_b,j}$ are also linear combinations of lower order terms,
\begin{equation}
 \begin{split}
  \frac{\partial \widetilde{Q}_{l_a,\mu_a,l_b,\mu_b,j}^{c/s,c/s}}{\partial R_{a,x}}
    = \quad{}&\frac{A_{l_a,\mu_a}}{A_{l_a-1,\mu_a+1}}\widetilde{Q}_{l_a-1,\mu_a+1,l_b,\mu_b,j}^{c/s,c/s}\\[5pt]
    - & \frac{A_{l_a,\mu_a}}{A_{l_a-1,\mu_a-1}}\widetilde{Q}_{l_a-1,\mu_a-1,l_b,\mu_b,j}^{c/s,c/s}\\[5pt]
    - & \frac{A_{l_b,\mu_b}}{A_{l_b-1,\mu_b+1}}\widetilde{Q}_{l_a,\mu_a,l_b-1,\mu_b+1,j}^{c/s,c/s}\\[5pt]
    + & \frac{A_{l_b,\mu_b}}{A_{l_b-1,\mu_b-1}}\widetilde{Q}_{l_a,\mu_a,l_b-1,\mu_b-1,j}^{c/s,c/s},
 \end{split}
  \label{eq:cartesian_dev_x}
\end{equation}
\begin{equation}
 \begin{split}
  \frac{\partial \widetilde{Q}_{l_a,\mu_a,l_b,\mu_b,j}^{c/s,c/s}}{\partial R_{a,y}}
    = \quad{}& (\pm 1)_{m_a}\frac{A_{l_a,\mu_a}}{A_{l_a-1,\mu_a+1}}\widetilde{Q}_{l_a-1,\mu_a+1,l_b,\mu_b,j}^{s/c,c/s}\\[5pt]
    + &(\pm 1)_{m_a}\frac{A_{l_a,\mu_a}}{A_{l_a-1,\mu_a-1}}\widetilde{Q}_{l_a-1,\mu_a-1,l_b,\mu_b,j}^{s/c,c/s}\\[5pt]
    - &(\pm 1)_{m_b}\frac{A_{l_b,\mu_b}}{A_{l_b-1,\mu_b+1}}\widetilde{Q}_{l_a,\mu_a,l_b-1,\mu_b+1,j}^{c/s,s/c}\\[5pt]
    - &(\pm 1)_{m_b}\frac{A_{l_b,\mu_b}}{A_{l_b-1,\mu_b-1}}\widetilde{Q}_{l_a,\mu_a,l_b-1,\mu_b-1,j}^{c/s,s/c},
 \end{split}
  \label{eq:cartesian_dev_y}
\end{equation}
\begin{equation}
 \begin{split}
  \frac{\partial \widetilde{Q}_{l_a,\mu_al_b,\mu_b,j}^{c/s,c/s}}{\partial R_{a,z}}
    = \quad{}&2\frac{A_{l_a,\mu_a}}{A_{l_a-1,\mu_a}}\widetilde{Q}_{l_a-1,\mu_a,l_b,\mu_b,j}^{c/s,c/s}\\[5pt]
    - &2 \frac{A_{l_b,\mu_b}}{A_{l_b-1,\mu_b}}\widetilde{Q}_{l_a,\mu_a,l_b-1,\mu_b}^{c/s,c/s}.
 \end{split}
  \label{eq:cartesian_dev_z}
\end{equation}
where $(\pm1)_m=1$ if $m\geq0$ and $(\pm1)_m=-1$ if $m<0$. Note that the cosine part of the $y$ derivatives are constructed from the sine part and vice versa. Furthermore, the terms in Equations~\eqref{eq:cartesian_dev_x}-\eqref{eq:cartesian_dev_z} with $l_{a/b}-1<0$ are zero. A special case has to be considered for 
the $x,y$ derivatives, when $\mu=0$. The matrix  elements of the type $\widetilde{Q}_{l_a-1,-1,l_b,\mu_b,j}^{c/s,c/s}$ and $\widetilde{Q}_{l_a,\mu_a,l_b-1,-1,j}^{c/s,c/s}$ are required for the construction of the $x$ and $y$ derivatives if $\mu_{a/b}=0$, see Equations~\eqref{eq:cartesian_dev_x} and \eqref{eq:cartesian_dev_y}. These matrix elements are never calculated 
since $\mu$ is positive by definition, but they can be obtained using the symmetry relations given in Equation~\eqref{eq:symmetry_relations_W}. For example if 
$\mu_a=0$, the following relations are used for the $x$-derivative
\begin{equation}
 \widetilde{Q}_{l_a-1,-1,l_b,\mu_b,j}^{cc/cs} = (-1)\widetilde{Q}_{l_a-1,1,l_b,\mu_b,j}^{cc/cs}
\end{equation}
and for the $y$ derivative we employ the symmetry relations:
\begin{equation}
 \widetilde{Q}_{l_a-1,-1,l_b,\mu_b,j}^{sc/ss} = \widetilde{Q}_{l_a-1,1,l_b,\mu_b,j}^{sc/ss}.
\end{equation}

%---------------------- Implementation Details
\section{Implementation Details}
\label{sec:imp_details}
Integrals of the type $(a|\mathcal{O}|b)$ have been implemented for the overlap $\delta(\mathbf{r})$, Coulomb $1/r$, long-range Coulomb $\erf(\omega r)/r$, 
short-range Coulomb $\erfc(\omega r)/r$, Gaussian-damped Coulomb $\exp(-\omega r^2)/r$ operator and the Gaussian operator $\exp(-\omega r^2)$, where 
$r=|\mathbf{r}_1-\mathbf{r}_2|$. The procedure for calculating these integrals differs only by the evaluation of the $s$-type integrals 
$(0_a|\mathcal{O}|0_b)$ and their derivatives with respect to $R_{ab}^2$. The expressions for the $k$-th derivatives $(0_a|\mathcal{O}|0_b)^{(k)}$ have been derived from Ref.~\onlinecite{Ahlrichs2006} and are explicitly given in Table~S1, see SI.\par
\begin{figure}
 \includegraphics[width=\linewidth]{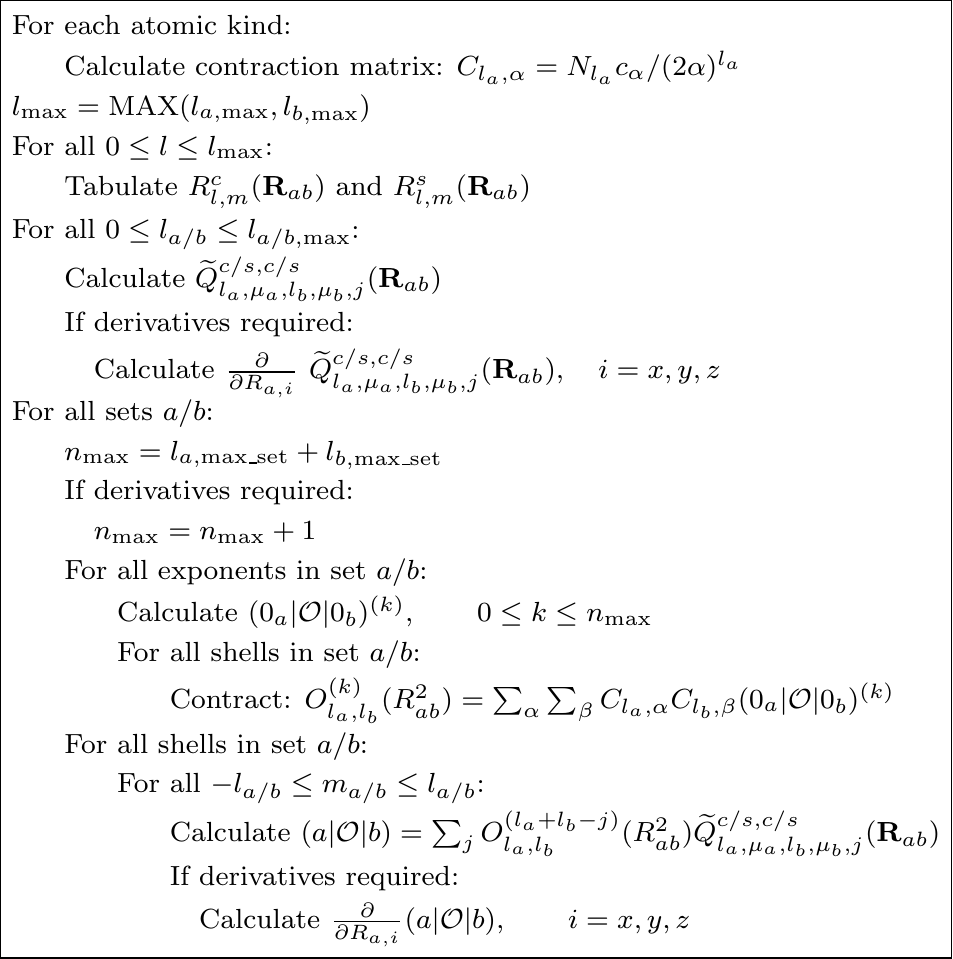}% Here is how to import EPS art
\caption{\label{fig:imp_details} Pseudocode for the calculation of the $(a|\mathcal{O}|b)$ integrals for an atom pair using a basis set with several sets of 
Gaussian functions as input. All functions that belong to one set share the same Gaussian exponents. Each set consists of shells characterized by the $l$ 
quantum number and a set of contraction coefficients.  }
\end{figure}
The pseudocode for the implementation of the SHG integrals is shown in Figure~\ref{fig:imp_details}. Our implementation is optimized for the typical structure 
of a Gaussian basis set, where Gaussian functions that share the same primitive exponents are organized in so-called sets. Since the matrix elements 
$\widetilde{Q}_{l_a,\mu_a,l_b,\mu_b,j}^{c/s,c/s}$ and their Cartesian derivatives do not depend on the exponents, they are computed only once for all 
$l=0,...,l_{\mathrm{max}}$, where $l_{\mathrm{max}}$ is the maximal $l$ quantum number of the basis set. The matrix elements 
$\widetilde{Q}_{l_a,\mu_a,l_b,\mu_b,j}^{c/s,c/s}$  are used multiple times for all functions with the same $l$ and $m$ quantum number. The integral and scalar 
derivatives $(0_a|\mathcal{O}|0_b)^{(k)}$ are then calculated for each set of exponents and subsequently contracted in one step using matrix-matrix 
multiplications. The same contracted monopole and its derivatives $O_{l_a,l_b}^{(k)}(R_{ab}^2)$ are used for all those functions with the same set of 
exponents and contraction coefficients, but different angular dependency $m$.\par
The only difference for the implementation of the $(a|r_a^{2n}|b)$ and $(ab\tilde{a})$ integrals is the evaluation of the 
contracted monopole and its scalar derivatives. For the three-index overlap integrals $(ab\tilde{a})$ we have additionally to consider the CG expansion. The 
expansion coefficients are independent on the position of the Gaussians and are precalculated only once for all $(ab\tilde{a})$ integrals. The Gaunt 
coefficients $G^{L,l,\tilde{l}}_{M,m,\tilde{m}}$ are obtained by multiplying Equation~\eqref{eq:CG_expansion} by $Y_{L,M}(\theta,\phi)$ and integrating over 
the angular coordinates $\phi$ and $\theta$ of the spherical polar system. The allowed values for $L$ range in steps of 2 from $|l-\tilde{l}|$ to 
$l+\tilde{l}$. Note that not all terms with $-L\leq M \leq L$ in Equation~\eqref{eq:CG_expansion} give non-zero contributions. For $l,\tilde{l}\leq2$, the 
product of two spherical harmonics is expanded in no more than four terms. However, the number of terms increases with $l+\tilde{l}$. A detailed discussion of 
the properties of the Gaunt coefficients can be found in Ref.~\onlinecite{Homeier1996} and tabulated values for low-order expansions of 
real-valued spherical harmonics are given in Ref.~\onlinecite{Giese2016_book}.\par
To assess the performance of the SHG integrals, an optimized OS scheme\cite{Obara1986} has been implemented. In the OS scheme, we first compute the Cartesian 
primitive integrals recursively. Subsequently, the Cartesian integrals are contracted and transformed to SpHGs. An efficient sequence of vertical and 
horizontal recursive steps is used to enhance the performance of the recurrence procedure. For the integrals $(a|b)$, $(a|r_a^{2n}|b)$ and $(ab\tilde{a})$, the 
recursion can be performed separately for each Cartesian direction, which drastically reduces the computational cost for high angular momenta. The contraction and 
transformation is performed in one step using efficient matrix-matrix operations. The three-index overlap integrals $(ab\tilde{a})$
are computed as described in Section~\ref{sec:integrals_aba} by combining the two Cartesian Gaussian functions at center $\mathbf{R}_a$ into a new Cartesian 
function at $\mathbf{R}_a$.  
%---------------------- Computational Details
%
\section{Computational Details}
\begin{table}
\caption{Specifications for the basis sets used for the performance tests. Number of $s, p, d,f, g, h$ and $i$ functions and 
their contraction length $K$.}
\label{tab:basissets_specification}
\begin{ruledtabular}
\begin{tabular}{llc}
basis set name                    & functions        & $K$ \\\hline
\texttt{TESTBAS-L0}               & 5s               & 1,...,7\\
\texttt{TESTBAS-L1}               & 5p               & 1,...,7\\
\texttt{TESTBAS-L2}               & 5d               & 1,...,7\\
\texttt{TESTBAS-L3}               & 5f               & 1,...,7\\
\texttt{TESTBAS-L4}               & 5g               & 1,...,7\\
\texttt{TESTBAS-L5}               & 5h               & 1,...,7\\
\texttt{H-DZVP-MOLOPT-GTH}        & 2s1p             & 7\\
\texttt{O-DZVP-MOLOPT-GTH}        & 2s2p1d           & 7\\
\texttt{O-TZV2PX-MOLOPT-GTH}      & 3s3p2d1f         & 7\\
\texttt{Cu-DZVP-MOLOPT-SR-GTH}    & 2s2p2d1f         & 6\\
\texttt{H-LRI-MOLOPT-GTH}         & 10s9p8d6f        & 1\\
\texttt{O-LRI-MOLOPT-GTH}         &15s13p12d11f9g    & 1\\
\texttt{Cu-LRI-MOLOPT-SR-GTH}     &15s13p12d11f10g9h8i   & 1
\end{tabular}
\end{ruledtabular}
\end{table}
The OS and SHG integral scheme have been implemented in the CP2K\cite{cp2k,Hutter2014} program suite and are available as separate packages. The measurements 
of the timings have been performed on an Intel Xeon (Haswell) platform\footnote{Intel\textsuperscript{\textregistered} Xeon\textsuperscript{\textregistered} 
E5-2697v3/DDR 2133} using the Gfortran Version 4.9.2 compiler with highest possible optimization. Matrix-matrix multiplications are efficiently computed using 
Intel\textsuperscript{\textregistered} MKL LAPACK Version 11.2.1. \par
Empirical timings have been measured for the integrals $(a|\mathcal{O}|b)$,  $(a|r_a^{2n}|b)$ and $(ab\tilde{a})$ using the basis sets specified in 
Table~\ref{tab:basissets_specification}. The basis sets at centers $\mathbf{R}_a$ and $\mathbf{R}_b$ are chosen to be identical. The measurements have been 
performed for a series of test basis sets with angular momenta $L=0,...,5$ and contraction lengths $K=1,...,7$. For example, the specification 
(\texttt{TESTBAS-L1}, $K$=7) indicates that we have five contracted $p$ functions at both centers, where each contracted function is a linear combination of seven primitive Gaussians. Furthermore, timings have been measured for basis sets of the \texttt{MOLOPT} type\cite{VandeVondele2007} that are widely used for DFT calculations with CP2K, see SI for details. The \texttt{MOLOPT} basis sets contain highly contracted functions with shared 
exponents, i.e. they are so-called family basis sets. A full contraction over all primitive functions is used for all $l$ quantum numbers. 
For the $(ab\tilde{a})$ integrals, we use for the second function at center $\mathbf{R}_a$, $\varphi_{\tilde{l}_a,\tilde{m}_a}$, the corresponding 
\texttt{LRI-MOLOPT} basis sets, see Table~\ref{tab:basissets_specification}. The latter is an auxiliary basis set and contains uncontracted functions, as 
typically used for RI approaches.
%---------------------- Results and Discussion
\section{Results and Discussion}
\label{sec:results_discussion}
This section compares the efficiency of the SHG scheme in terms of mathematical operations and empirical timings to the widely used OS method.
\subsection{Comparison of the algorithms}
Employing the OS scheme for the evaluation of SpHG integrals, the most expensive step is typically the recursive computation of the primitive Cartesian 
Gaussian integrals. The recurrence procedure is increasingly demanding in terms of computational cost for large angular momenta. The recursion depth is even 
increased when the gradients of the integrals are required, since the derivatives of Cartesian Gaussian functions are 
constructed from higher-order angular terms $(l+1)$. In case of the $\texttt{TESTBAS-L5}$ basis set, the computational cost for evaluating both, the 
Coulomb integral $(a|1/r|b)$ and its derivatives, is three times larger than for calculating solely the integral. 
The integral matrix of primitive Cartesian integrals (and their derivatives) has to be transformed to primitive SpHG integrals, which are then 
contracted. The contribution of the contraction step to the total computational cost is small for integrals with non-local 
operators. However, the OS recursion takes a significantly smaller amount of time for local operators, when efficiently implemented, see 
Section~\ref{sec:imp_details}. Thus, the contraction of the primitive SpHG integrals contributes by up to 50\% to the total timings for the 
integrals $(ab)$, $(ab\tilde{a})$ and $(a|r_a^{2n}|b)$. The contraction step can be even dominant when derivatives of these integrals are required since it has 
to be performed for each spatial direction, i.e. we have to contract the $x,y$ and $z$ Cartesian derivatives of the primitive integral matrix separately. 
Details on the contribution of the different steps to the overall computational cost are displayed in Figures~S1-S4 (a,b), SI.   \par
\begin{table}
\caption{ Number of matrix elements that need to be contracted for two-index integrals comparing the OS and SHG method for integral (Int.) and integral+derivative (Int.+Dev.) evaluation. The basis set specifications are given in Table~\ref{tab:basissets_specification}.}
\label{tab:operations_contraction}
\begin{ruledtabular}
\begin{tabular}{lCCCCCCCC}
  \multirow{2}{*}{Integral method}&\multicolumn{2}{c}{H-DZVP}& \multicolumn{2}{c}{O-DZVP} \\\cmidrule(l{0.1em}r{0.1em}){2-3} 
\cmidrule(l{0.1em}r{0.1em}){4-5}
 &\textnormal{Int.}&\textnormal{Int.+Dev.}&\textnormal{Int.}&\textnormal{Int.+Dev.} \\ \hline
 OS  & 784  & 3136 & 3969  & 15876\\
 SHG & 147 & 196   & 245   & 294
\end{tabular}
\end{ruledtabular}
\end{table}
The SHG method requires only recursive operations for the evaluation of $R_{l,m}^{c/s}$ [Equations~\eqref{eq:regular_scaled_harmonics_recursion_start}-\eqref{eq:regular_scaled_harmonics_recursion}], which do not 
depend on the Gaussian exponents and can be tabulated for all functions of the basis set. Furthermore, a deeper recursion is not required for the derivatives 
of the integrals because they are constructed from linear combinations of lower-order angular terms, see 
Equations~\eqref{eq:cartesian_dev_x}-\eqref{eq:cartesian_dev_z}. Instead of contracting each primitive SpHG, we contract an auxiliary integral 
of $s$ functions and its scalar derivatives. The number of scalar derivatives is linearly increasing with $l$. If the gradients are required, the 
increase in computational cost for the contraction is marginal. We have to contract only one additional scalar derivative of the auxiliary integral. 
As shown in Table~\ref{tab:operations_contraction}, the number of matrix elements, which have to be contracted for the $\texttt{MOLOPT}$ basis sets, is 1-2 
orders of magnitude smaller for the SHG scheme. Note that the numbers of SHG matrix elements refer to our implementation, where actually more scalar derivatives 
of $(0_a|\mathcal{O}|0_b)$ and $(0_a|r_a^{2n}|0_b)$ are contracted than necessary, in order to enable library-supported matrix multiplications.\par
For both methods, we have to calculate the same number of fundamental integrals $(0_a|\mathcal{O}|0_b)$ and their scalar derivatives with respect to $R_{ab}^2$ 
(SHG) and $-\rho R_{ab}^2$ (OS),\cite{Ahlrichs2006} where $\rho=\alpha\beta/(\alpha+\beta)$. The time for 
evaluating these auxiliary integrals is approximately the same for both methods. In the SHG scheme, the evaluation of the latter constitutes the major 
contribution to the total timings for
highly contracted basis sets with different sets of exponents. The remaining operations are orders of magnitudes faster than those in the OS scheme.
Details are given in Figures~S1-S4 (c,d). The recursive procedure to obtain regular scaled solid harmonics is negligible in terms of computational cost. The evaluation of 
$\widetilde{Q}^{c/s,c/s}_{l_a,\mu_a,l_b,\mu_b,j}$ [Equation~\eqref{eq:Q_tilde}] from the pretabulated $R_{l,m}^{c/s}$ contributes increasingly
for large angular momenta. The integrals $(a|\mathcal{O}|b)$ are finally constructed from the contracted quantity $O^{(k)}_{l_a,l_b}$ [Equation~\eqref{eq:O_rab2_deriv}] and 
$\widetilde{Q}^{c/s,c/s}_{l_a,\mu_a,l_b,\mu_b,j}$ as displayed Figure~\ref{fig:imp_details}. This step becomes increasingly expensive for large $l$ quantum numbers and is 
in fact dominant for family basis sets, where the fundamental integrals are calculated only for one set of exponents.  
\subsection{Speed-up with respect to the OS method}
\begin{figure}
 \includegraphics[width=\linewidth]{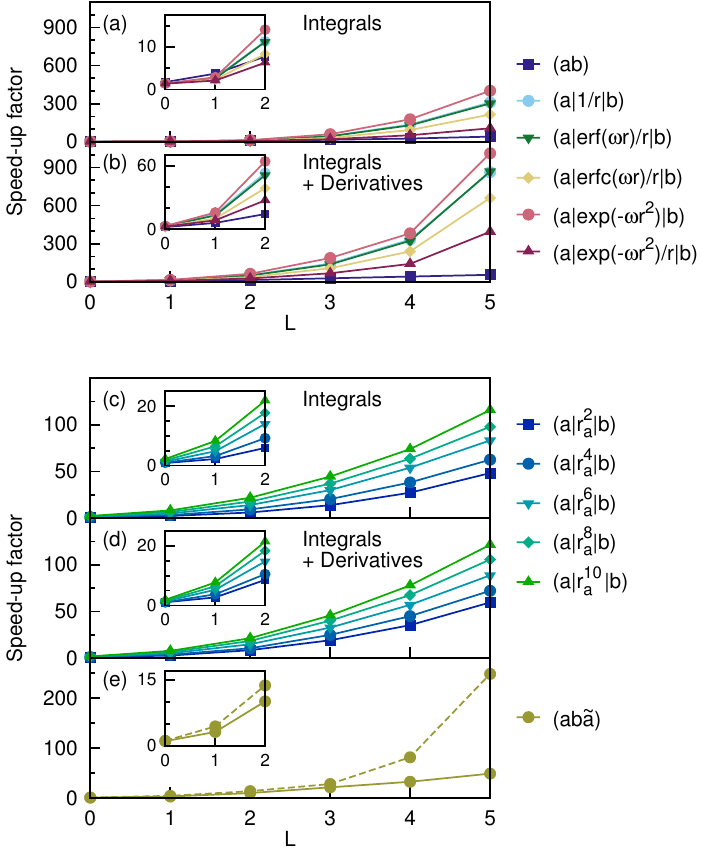}% Here is how to import EPS art
\caption{\label{fig:speedupK7} Speed-up for different two-center integrals dependent on the $l$ quantum number at the fixed contraction length $K=7$. The 
speed-up factor is defined as the ratio
OS/SHG. Speed-up for (a,b) integrals $(a|\mathcal{O}|b)$, (c,d) $(a|r_a^{2n}|b)$ and (e) $(ab\tilde{a})$. The solid line in (e) is the speed-up for the 
integrals and the dashed line the speed-up for both, integrals + derivatives.}
\end{figure}
\begin{figure}
 \includegraphics[width=\linewidth]{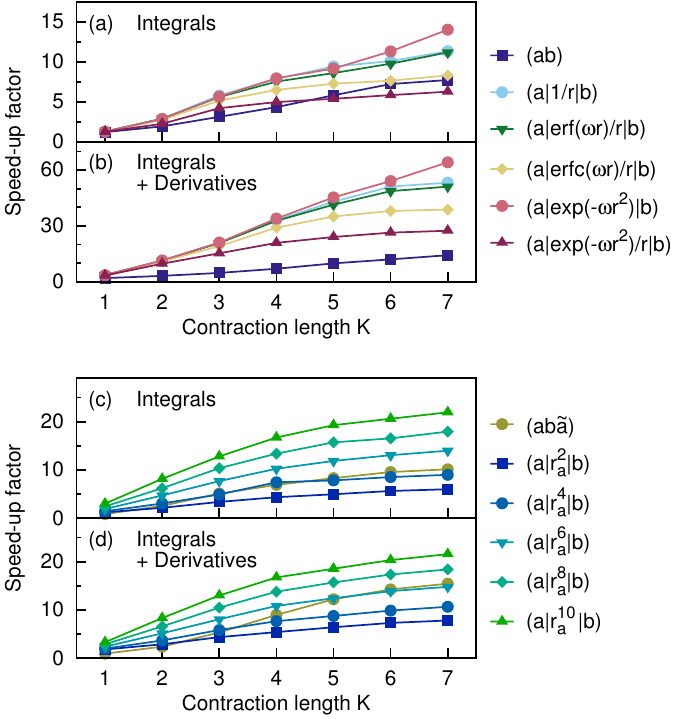}% Here is how to import EPS art
\caption{\label{fig:speedupL2} Speed-up for different two-center integrals dependent on the contraction length $K$. The $l$ quantum number is fixed and set to 
$l=2$. The speed-up factor is defined as the ratio OS/SHG. Speed-up for (a,b) integrals $(a|\mathcal{O}|b)$, (c,d) $(a|r_a^{2n}|b)$ and $(ab\tilde{a})$.}
\end{figure}
\begin{table*}
\caption{ Speed-up for different two-center integrals. The speed-up is defined as the ratio of the timings OS/SHG. The basis set specifications are given in 
Table~\ref{tab:basissets_specification}. 
}
\label{tab:realistic_basis_sets}
\begin{ruledtabular}
\begin{tabular}{lCCCCCCCC}
  \multirow{2}{*}{Integral type}&\multicolumn{2}{c}{H-DZVP}& \multicolumn{2}{c}{O-DZVP} & \multicolumn{2}{c}{O-TZV2PX} 
&\multicolumn{2}{c}{Cu-DZVP}\\\cmidrule(l{0.1em}r{0.1em}){2-3} 
\cmidrule(l{0.1em}r{0.1em}){4-5}\cmidrule(l{0.1em}r{0.1em}){6-7} \cmidrule(l{0.1em}r{0.1em}){8-9}
 &\textnormal{Int.}&\textnormal{Int.+Dev.}&\textnormal{Int.}&\textnormal{Int.+Dev.} &\textnormal{Int.}&\textnormal{Int.+Dev.} 
&\textnormal{Int.}&\textnormal{Int.+Dev.}\\ \hline
 $(ab)$                                & 2.4 & 2.4 & 6.2  & 5.5  & 11.4 &10.3 & 8.9  & 8.3\\[5pt]
 $(a|1/r|b)$                           & 1.8 & 6.2 & 5.9  & 18.4 & 16.8 &31.6 & 14.6 & 26.0 \\[5pt]
 $(a|\textnormal{erf}(\omega r)/r|b)$  & 1.7 & 6.0 & 5.8  & 18.4 & 16.6 &31.7 & 14.4 & 26.0 \\[5pt]
 $(a|\textnormal{erfc}(\omega r)/r|b)$ & 1.7 & 5.4 & 5.2  & 16.3 & 14.9 &29.5 & 12.9 & 24.8\\[5pt]
 $(a|\exp(-\omega r^2)|b)$             & 1.8 & 6.7 & 6.4  & 19.7 & 18.0 &32.5 & 16.0 &27.4\\[5pt]
 $(a|\exp(-\omega r^2)/r|b)$           & 1.6 & 5.0 & 4.4  & 14.1 & 12.3 &25.4 & 10.8 &22.0\\[5pt]
 $(a|r^2_a|b)$                         & 2.6 & 2.7 & 9.7  & 8.8  & 22.9 &18.6 & 19.7 & 15.8\\[5pt]
 $(a|r^4_a|b)$                         & 4.0 & 4.0 & 16.0 & 14.0 & 39.4 &29.3 &34.7  &25.2\\[5pt]
 $(a|r^6_a|b)$                         & 6.6 & 6.3 & 25.3 & 21.6 & 59.5 &44.3 & 56.1 &38.9\\[5pt]
 $(a|r^8_a|b)$                         & 9.1 & 8.1 & 34.7 & 29.6 &79.3  &61.4 &73.4  &54.6\\[5pt]
 $(a|r^{10}_a|b)$                      & 11.8& 10.5& 44.7 &36.7  &105.2 &79.9 &97.5  &72.2\\[5pt]
 $(ab\tilde{a})$                       & 7.0 & 7.6 & 10.1 & 8.7 & 7.5 & 7.2 & 7.2 & 10.5
 
\end{tabular}
\end{ruledtabular}
\end{table*}
Figure~\ref{fig:speedupK7} displays the performance of the SHG scheme as function of the $l$ quantum number. The speed-up gained by the SHG method is presented 
for the basis sets \texttt{TESTBAS\_LX} for a fixed contraction length. Generally, the ratio of the timings OS/SHG increases with increasing $l$. For the 
$(a|\mathcal{O}|b)$ integrals, we observe speed-ups between 40 and 400 for $l=5$. For $s$ functions, our method can become up to a 
factor of two faster. The smallest speed-up is obtained for the overlap integrals since the OS recursion can be spatially separated. The speed-up for the other operators depends on the computational cost for the evaluation of the primitive Gaussian integrals $(0_a|\mathcal{O}|0_b)$. The SHG 
method outperforms the OS 
scheme by up to a factor of 1000 ($l=5$) when also the derivatives of $(a|\mathcal{O}|b)$ are computed.\par
The computational cost for calculating $(a|r_a^{2n}|b)$ integrals of $h$ functions is up to two orders of magnitude reduced compared to the OS scheme. The 
speed-up increases with $n$. The SHG method is beneficial for all $l>0$ and also for $l=0$ when $n\geq 3$. The speed-up factor is generally slightly 
larger when also the derivatives are required. However, the performance increase is not as pronounced as for the derivatives of $(a|\mathcal{O}|b)$ which is 
again due to the efficient spatial separation of the OS recurrence. \par 
The performance improvement for $(ab\tilde{a})$ is comparable to the 
$(a|r_a^{2n}|b)$ integrals. For the derivatives of $(ab\tilde{a})$ on the other hand, we get a significantly larger speed-up due to the fact that it increases 
more than linearly with $l$ and that the  OS recurrence has to be performed for larger angular momenta. For instance, the derivatives of the 
$h$ functions require the recursion up to $l_a+l_{\tilde{a}}+1=11$.\par
Figure~\ref{fig:speedupL2} shows the performance of the SHG scheme as function of the contraction length $K$. The speed-up increases with $K$ for 
all integral types. A saturation is observed around $K=6,7$ for $(a|r_a^{2n}|b)$ and some of the $(a|\mathcal{O}|b)$ integrals, for example $(a|1/r|b)$. 
The reason is that the computation of the fundamental integrals $(0_a|1/r|0_b)^{(k)}$ increasingly contributes with $K$ to the total computational cost in the SHG scheme,
whereas its relative contribution to the total time is approximately constant in the OS scheme, see Figure~S3 (SI). For $K=7$ and $l=2$, the evaluation of $(0_a|1/r|0_b)^{(k)}$ 
is with 70\% the predominant step in the SHG scheme. Since the absolute time for calculating the fundamental integrals is the same in both schemes, the 
increase in speed-up levels off.  
The saturation effect is less pronounced, for example, for the overlap $(ab)$ because the evaluation of $(0_a0_b)^{(k)}$ is computationally less expensive than 
for $(0_a|1/r|0_b)^{(k)}$. Its relative contribution to the
total time in the SHG scheme is with 50\% significantly smaller, see Figure~S3 (d) for $K=7$.
However, the saturation for large $K$ is hardly of practical relevance because the contraction lengths of Gaussian basis sets is typically not larger than $K=7$.\par
The speed-up for separate operations in the integral evaluation can only be assessed for steps such as the contraction, which have an equivalent in the OS 
scheme. The SHG contraction is increasingly beneficial 
for large $l$ quantum numbers, large contraction lengths and when also derivatives are computed, see Figure~S5 (SI).\par
Table~\ref{tab:realistic_basis_sets} presents the performance of the SHG method for the \texttt{MOLOPT} basis sets. We find that the SHG scheme is superior to 
the OS method for all two-center integrals and basis sets. The smallest performance enhancement is obtained for the \texttt{DZVP} basis set of 
hydrogen, where we get a speed-up by a factor of 1.5-10 because only $s$ and $p$ functions are included in this basis set. A 
performance improvement of 1-2 orders of magnitude is observed for the basis sets that include also $f$ functions. 
The largest speed-up is obtained for the $(a|r_a^{2n}|b)$ integrals followed by the Coulomb and modified Coulomb integrals. The 
SHG scheme is even more beneficial, at least for $(a|\mathcal{O}|b)$ integrals, when also the derivatives are computed. 
For the integrals $(ab)$, $(a|r_a^{2n}|b)$ and $(ab\tilde{a})$  on the contrary, the speed-up for both, integrals and derivatives, is instead a bit smaller than for the calculation of the integrals alone.  
This behavior has to be related to the fact that the \texttt{MOLOPT} basis sets are family basis sets. The OS recursion is carried out for only one set of exponents. Therefore, this part of the calculation is
computationally less expensive than for basis sets constituted of several sets of exponents. Furthermore, the OS recursion is computationally less demanding for integrals with local operators, see Section~\ref{sec:imp_details}, and the
computational cost for the recursion is in this case only slightly increased when additionally computing the derivatives. In the SHG scheme, the construction of the derivatives from the
contracted quantity given in Equation~\eqref{eq:O_rab2_deriv} and $\widetilde{Q}^{c/s,c/s}_{l_a,\mu_a,l_b,\mu_b,j}$ [Equations~\eqref{eq:Q_tilde}] and its derivatives 
[Equations~\eqref{eq:cartesian_dev_x}-\eqref{eq:cartesian_dev_z}] is the dominant step for family basis sets.
This construction step cannot be supported by memory-optimized library routines and the relative increase in computational cost upon calculating the derivatives is in this particular case larger than for the OS scheme. \par
For the computation of molecular integrals in quantum chemical simulations, the relation $(a|\mathcal{O}|b)=(-1)^{l_b-l_a}(b|\mathcal{O}|a)$ can be employed if we have the same set of functions at centers
$\mathbf{R}_a$ and $\mathbf{R}_b$. This relation has not been used for the measurements of the empirical timings, but is in practice useful when the atoms at center $\mathbf{R}_a$ and $\mathbf{R}_b$ are of the same 
elemental type.
%---------------------- Conclusion
\section{Conclusions}
Based on the work of Giese and York\cite{Giese2008}, we used Hobson's theorem to derive expressions for the SHG integrals $(a|r_a^{2n}|b)$ and $(ab\tilde{a})$ 
and their derivatives. We showed that the SHG overlap $(ab\tilde{a})$ is a sum of $(a|r_a^{2n}|b)$ integrals. Additionally, two-center SHG integrals with Coulomb, modified Coulomb and Gaussian operators have been implemented 
adapting the expressions given in Refs.~\onlinecite{Ahlrichs2006} and \onlinecite{Giese2008}.\par
In the SHG integral scheme, the angular-dependent part is separated from the exponents of the Gaussian primitives. As a consequence, the contraction is only 
performed for $s$-type auxiliary integrals and their scalar derivatives. The angular-dependent term is obtained by a relatively simple recurrence procedure 
and can be pre-computed. In contrast to the Cartesian Gaussian-based OS scheme, the derivatives with respect to the spatial 
directions are computed from lower-order 
$(l-1)$ terms. \par
We showed that the SHG integral method is superior to the OS scheme by means of empirical timings. Performance improvements have been 
observed for all integral types, in particular for higher angular momenta and high 
contraction lengths. Specifically for the 
$(a|r_a^{2n}|b)$ integrals, the timings ratio OS/SHG grows with increasing $n$. The speed-up is usually even larger for 
the computation of the Cartesian derivatives. This is especially true for Coulomb-type integrals. 
%------------------------ Acknowledgement
\begin{acknowledgments}
We thank Andreas Gl\"o{\ss} for helpful discussions and technical support.
J. Wilhelm thanks the NCCR MARVEL, funded by the Swiss National Science Foundation, for financial support. N. Benedikter acknowledges support by ERC Advanced grant 321029 and by VILLUM
FONDEN via the QMATH Centre of Excellence (Grant No. 10059).
\end{acknowledgments}
%
%---------------------- Supplementary
\section*{Supplementary Information}
 Supplementary Material is available for the analytic expressions of $(0_a|\mathcal{O}|0_b)^{(k)}$ employing the standard Coulomb, modified Coulomb and 
Gaussian-type operators, see Table~S1. Further information on integral timings is presented in Figures~S1-S5. A detailed description of the \texttt{MOLOPT} basis set is given in Tables~S2-S8.

%----------------------------Appendix
\appendix

%
%%%%%%%%%%%%%%%%%%%%%%%%%%%%%%%%%%%%%% Proof of general formal for \chi_{l,m}(alpha,\mathbf{r}_a)r_a^{2n}
\section{Proof of general formula for \boldmath{$\chi_{l,m}(\alpha,\mathbf{r}_a)r_a^{2n}$}}
\label{app:proof_product}
In this appendix, we prove that Equation~\eqref{eq:generic_chira2n}
is valid for all $ n\in\mathbb{N}$. In the following, the label $tbs$ indicates that the identity of the left-hand side (lhs) and the right-hand side (rhs) of 
the equation remains 
\textit{to be shown}.
%--------------def product
\begin{definition}
 The product of a solid harmonic Gaussian function at center $\mathbf{R}_a$ multiplied with the operator $r_a^{2n}$ is defined as
 \begin{equation}
   \chi_{l,m}(\alpha,\mathbf{r}_a)r_a^{2n} :=  C_{l,m}(\mathbf{r}_a) \exp{\left(-\alpha r_a^2\right)} r_a^{2n},
 \end{equation}
 where  $C_{l,m}$ is the solid harmonic defined in Equation~\eqref{eq:solid_harominc_C} and $n\in\mathbb{N}$.
 \label{def:product}
\end{definition}
%------------------------------Theorem 1
Recall that $C_{l,m}(\nabla_a)$ is the spherical tensor gradient operator (STGO) acting on center $\mathbf{R}_a$. In the following we generally drop all 
`passive' indices, writing e.\,g.\ $\RN{1}_n$ instead of $\RN{1}_n^{l,m,\alpha,\mathbf{r}_a}$.
\begin{theorem}\label{lemma:In}
 Equation~\eqref{eq:generic_chira2n}, 
 \begin{equation}\label{eq:theoremstatement}
 \begin{split}
   \chi_{l,m}(\alpha,\mathbf{r}_a)r_a^{2n} = \frac{C_{l,m}(\nabla_a) }{(2\alpha)^{l}}\sum_{j=0}^n&\binom{n}{j} 
\frac{(l+j-1)!}{(l-1)!\alpha^j}\\&\times\exp{\left(-\alpha r_a^2\right)} r_a^{2(n-j)}.
 \end{split}
\end{equation}
 is valid for all $n\in\mathbb{N}$ .
\end{theorem}
 \begin{proof}
 Using Hobson's theorem\cite{Hobson1892} yields
\begin{equation}
 \begin{split}
 C_{l,m}&(\nabla_a) \exp{\left(-\alpha r_a^2\right)}r_a^{2n}\\
&=(-2)^{l}C_{l,m}(\mathbf{r}_a)\left(\frac{d}{dr_a^2}\right)^{l}\left[\exp{\left(-\alpha r_a^2\right)} r_a^{2n}\right].
 \end{split}
 \label{eq:hobson_theorem_proof}
\end{equation}
By applying Leibniz's rule of differentiation we get
\begin{equation}
 \begin{split}
   &\left(\frac{d}{dr_a^2}\right)^l\exp{\left(-\alpha r_a^2\right)} r_a^{2n} \\&= \sum_{j=0}^{\mathrm{min}(l,n)}\!\binom{l}{j}\! 
\left[\left(\frac{d}{dr_a^2}\right)^{l-j}\exp{(-\alpha r_a^2)}\right]\left[\left(\frac{d}{dr_a^2}\right)^{j} \left(r_a^2\right)^n\right]\\
  &=\sum_{j=0}^{\mathrm{min}(l,n)} \binom{l}{j}(-\alpha)^{l-j} \exp{(-\alpha r_a^2)} \frac{n!}{(n-j)!} \left(r_a^2\right)^{n-j}\\
  &=\sum_{j=0}^{\mathrm{min}(l,n)}\binom{n}{j}\frac{l!}{(l-j)!}(-\alpha)^{l-j} \exp{(-\alpha r_a^2)}  \left(r_a^2\right)^{n-j}.
 \end{split}
 \label{eq:after_leibniz}
\end{equation}
Inserting the last line of Equation~\eqref{eq:after_leibniz} in Equation~\eqref{eq:hobson_theorem_proof} and writing out the term for $j=0$ explicitly leads to
\begin{equation}
 \begin{split}
  C_{l,m}&(\nabla_a) \exp{\left(-\alpha r_a^2\right)}r_a^{2n}\\
  ={}&(-2)^{l}(-\alpha)^l \exp{(-\alpha r_a^2)}  \left(r_a^2\right)^{n}C_{l,m}(\mathbf{r}_a)\\
  & + (-2)^{l}\sum_{j=1}^{\mathrm{min}(l,n)}\binom{n}{j}\frac{l!}{(l-j)!}(-\alpha)^{l-j}\\&\quad\times \exp{(-\alpha r_a^2)}  
\left(r_a^2\right)^{n-j}C_{l,m}(\mathbf{r}_a).
 \end{split}
\end{equation}
Employing Definition~\ref{def:product} and solving for $\chi_{l,m}(\alpha,\mathbf{r}_a)r_a^{2n}$ we obtain
\begin{equation}
 \begin{split}
 \chi_{l,m}&(\alpha,\mathbf{r}_a)r_a^{2n}\\={}&\frac{C_{l,m}(\nabla_a) }{(2\alpha)^{l}} \exp{\left(-\alpha r_a^2\right)} 
r_a^{2n}-\sum_{j=1}^{\mathrm{min}(l,n)}\binom{n}{j}\frac{l!}{(l-j)!}\\&\times(-\alpha)^{-j} \exp{(-\alpha r_a^2)}  
\left(r_a^2\right)^{n-j}C_{l,m}(\mathbf{r}_a).
 \end{split}
\end{equation}
Introducing the notation
\begin{equation}\label{eq:defIn}
 \RN{1}_n:=\chi_{l,m}(\alpha,\mathbf{r}_a)r_a^{2n}
\end{equation}
and recalling Definition \ref{def:product}, we obtain a recursion relation:
\begin{equation}\label{eq:recursion}
 \begin{split}
 \RN{1}_n = & \frac{C_{l,m}(\nabla_a) }{(2\alpha)^{l}} \exp{\left(-\alpha r_a^2\right)} 
r_a^{2n}\\&-\sum_{j=1}^{\mathrm{min}(l,n)}\binom{n}{j}\frac{l!}{(l-j)!}(-\alpha)^{-j}\RN{1}_{n-j}.
 \end{split}
\end{equation}
Furthermore, it is easy to see (applying Hobson's theorem as done above for general $n$) that
\begin{equation}
\RN{1}_0 = \frac{C_{l,m}(\nabla_a) }{(2\alpha)^{l}} \exp{\left(-\alpha r_a^2\right)}. 
\label{eq:recursionstart}
\end{equation}
From here, the theorem can in principle be obtained by using \eqref{eq:recursion} and \eqref{eq:recursionstart} recursively. This is made mathematically 
rigorous by an induction proof in Lemma \ref{lemma:induction} below.
\end{proof}
%-------------------------------- Lemma 2
Let us denote the rhs of \eqref{eq:theoremstatement} by $\RN{2}_n$,
\begin{equation}
 \RN{2}_n:=\frac{C_{l,m}(\nabla_a)}{(2\alpha)^{l}}\sum_{j=0}^n\binom{n}{j} 
\frac{(l+j-1)!}{(l-1)!\alpha^j}\exp{\left(-\alpha r_a^2\right)} r_a^{2(n-j)}.
\label{eq:product_chira2ndef}
\end{equation}
The following Lemma tells us that the recursive representation \eqref{eq:recursion}-\eqref{eq:recursionstart} indeed has its closed form
given by $\RN{2}_n$.
\begin{lemma}\label{lemma:induction}
For all $n \in \mathbb{N}$ we have
\begin{equation}
 \RN{1}_n=\RN{2}_n.
\end{equation}
\end{lemma}
\begin{proof}
This is proved by mathematical induction.\\
1. Basis:  Recalling \eqref{eq:recursionstart}, it obviously holds $\RN{1}_0=\RN{2}_0$.\\
2. Induction Hypothesis: we assume that $\RN{1}_i = \RN{2}_i$ for all natural numbers $i<n$.\\
3. Inductive Step: We use the recursion relation \eqref{eq:recursion}. Since we sum over $j\geq1$, we can use the induction hypothesis $\RN{1}_{n-j} = 
\RN{2}_{n-j}$ to get 
\begin{equation}
 \begin{split}
 \RN{1}_n ={}&\frac{C_{l,m}(\nabla_a) }{(2\alpha)^{l}} \exp{\left(-\alpha r_a^2\right)} 
r_a^{2n}\\&-\sum_{j=1}^{\mathrm{min}(l,n)}\binom{n}{j}\frac{l!}{(l-j)!}(-\alpha)^{-j}\RN{2}_{n-j}.
 \end{split}
\end{equation}
Inserting the definition of $\RN{2}_{n-j}$ (i.\,e.\ Equation~\eqref{eq:product_chira2ndef} for $n-j$) this becomes
\begin{equation}
 \begin{split}
  \RN{1}_n ={}& \frac{C_{l,m}(\nabla_a) }{(2\alpha)^{l}} \exp{\left(-\alpha r_a^2\right)} 
r_a^{2n}\\&-\sum_{j=1}^{\mathrm{min}(l,n)}\binom{n}{j}\frac{l!}{(l-j)!}(-\alpha)^{-j}\frac{C_{l,m}(\nabla_a) 
}{(2\alpha)^{l}}\\&\times\sum_{k=0}^{n-j}\binom{n-j}{k} \frac{(l+k-1)!}{(l-1)!\alpha^k}\exp{\left(-\alpha r_a^2\right)} r_a^{2(n-j-k)}.
 \end{split}
 \label{eq:general_formula_complex}
\end{equation}
In the following, it is shown that Equation~\eqref{eq:general_formula_complex} is indeed equal to $\RN{2}_n$. 
All terms of $\RN{1}_n$ with $j>0$ in Equation~\eqref{eq:general_formula_complex} are denoted by 
\begin{equation}
 \begin{split}
 \RN{1}_n' := &-\sum_{j=1}^{\mathrm{min}(l,n)}\binom{n}{j}\frac{l!}{(l-j)!}(-\alpha)^{-j}\frac{C_{l,m}(\nabla_a) 
}{(2\alpha)^{l}}\\&\times\sum_{k=0}^{n-j}\binom{n-j}{k} \frac{(l+k-1)!}{(l-1)!\alpha^k}\exp{\left(-\alpha r_a^2\right)} r_a^{2(n-j-k)}
 \end{split}
\end{equation}
and the contributions with $j>0$ to $\RN{2}_n$ in Equation~\eqref{eq:product_chira2ndef} are in the following referred to as
\begin{equation} 
 \RN{2}_n' := \frac{C_{l,m}(\nabla_a)}{(2\alpha)^{l}}\sum_{j=1}^n\binom{n}{j}\frac{(l+j-1)!}{(l-1)!\alpha^j}\exp{\left(-\alpha r_a^2\right)} r_a^{2(n-j)}.
\end{equation}
To prove that $\RN{1}_n=\RN{2}_n$, it is sufficient to show that $\RN{1}_n'=\RN{2}_n'$. Both sides are reduced to
\begin{equation}
 \begin{split}
  -\sum_{j=1}^{\mathrm{min}(l,n)}&\binom{n}{j}\frac{l!}{(l-j)!}(-\alpha)^{-j}\\&\times\sum_{k=0}^{n-j}\binom{n-j}{k}
\frac{(l+k-1)!}{\alpha^k} r_a^{2(n-j-k)}\\\stackrel{\mathrm{tbs}}{=}&\sum_{j=1}^n\binom{n}{j}\frac{(l+j-1)!}{\alpha^j} 
r_a^{2(n-j)}
 \label{eq:lhs_rhs}
 \end{split}
\end{equation}
where we denote the lhs by
\begin{equation}
 \begin{split}
 \RN{1}_n'' :={}& -\sum_{j=1}^{\mathrm{min}(l,n)}\binom{n}{j}\frac{l!}{(l-j)!}(-\alpha)^{-j}\\&\times\sum_{k=0}^{n-j}\binom{n-j}{k}   
  \frac{(l+k-1)!}{\alpha^k} r_a^{2(n-j-k)}.
 \end{split}
\end{equation}
In order to sort by the exponents of $r_a^2$ in expression $\RN{1}_n''$, the Kronecker delta $\delta_{m,j+k}$ is introduced.
\begin{equation}
 \begin{split}
\RN{1}_n''={}&\sum_{m=1}^{n}r^{2(n-m)}_a\Biggl(-\sum_{j=1}^{\mathrm{min}(l,m)}\binom{n}{j}\frac{l!}{(l-j)!}(-\alpha)^{
-j}\\[5pt]&\times\sum_{k=0}^{n-j}\binom{n-j}{k}\frac{(l+k-1)!}{\alpha^k} \delta_{m,j+k}\Biggr).
 \end{split}
\end{equation}
The range of the newly introduced sum is $m=1,...,n$ since for the lower bound of summation we find that $m=j+k\geq1+0=1$ and for the upper bound 
$m=j+k\leq j+(n-j)=n$. For the inner sum over indices $j$, it must be considered that $k=m-j$ is negative if $j>m$ while the lower bound of the $k$-sum is 
in fact $k\geq0$. Thus, the upper range of the summation of the $j$-sum has to be changed to $\mathrm{min}(l,n,m)$, which is equivalent to $\mathrm{min}(l,m)$ 
because $m \leq 
n$. The summation ranges for the innermost sum are not modified since $k=m-j\leq n-j$. In the next step, the $k$-sum is eliminated replacing $k$ by $m-j$,
\begin{equation}
 \begin{split}
 \RN{1}_n''={}&
\sum_{m=1}^{n}\Biggr( -\sum_{j=1}^{\mathrm{min}(l,m)}\binom{n}{j}\frac{l!}{(l-j)!}(-\alpha)^{-j}\\[5pt]&\times
   \binom{n-j}{m-j}\frac{(l+m-j-1)!}{\alpha^{m-j}}\Biggr) r^{2(n-m)}_a.
 \end{split}
\end{equation}
Renaming the summation index on the rhs of Equation~\eqref{eq:lhs_rhs}, we get
\begin{equation}
 \RN{2}_n'':=\sum_{m=1}^n \binom{n}{m} \frac{(l+m-1)!}{\alpha^m}r^{2(n-m)}_a.
\end{equation}
We are done if we can show that $\RN{1}_n''=\RN{2}_n''$. We do this by comparing summand by summand, i.\,e.\ we have to show that for each $m=1,\ldots, n$,
\begin{equation}
 \begin{split}
    -\sum_{j=1}^{\mathrm{min}(l,m)}&\binom{n}{j}\frac{l!}{(l-j)!}(-\alpha)^{-j}
   \binom{n-j}{m-j}\\&\times\frac{(l+m-j-1)!}{\alpha^{m-j}}\\\stackrel{\mathrm{tbs}}{=}{}&\binom{n}{m} \frac{(l+m-1)!}{\alpha^m}.
 \end{split}
\end{equation}
Expansion of the binomial coefficients and further reduction gives
\begin{equation}
 \sum_{j=1}^{\mathrm{min}(l,m)}(-1)^{j+1}\binom{l}{j} \frac{(l+m-j-1)!}{(m-j)!}\stackrel{\mathrm{tbs}}{=} \frac{(l+m-1)!}{m!}.
 \label{eq:binomial_sum_jsplit}
\end{equation}
The term on the rhs is in fact the negative of the `missing' summand $j=0$ on the lhs and thus we have
\begin{equation}
  \sum_{j=0}^{\mathrm{min}(l,m)}(-1)^{j+1}\binom{l}{j} \frac{(l+m-j-1)!}{(m-j)!}\stackrel{\mathrm{tbs}}{=}0.
 \label{eq:binomial_sum_induction}
\end{equation}
The lhs is indeed zero, which is easily rationalized by dividing Equation~\eqref{eq:binomial_sum_induction} by $(-1)$ and assuming that $\mathrm{min}(l,m)=l$,
\begin{equation}
 \sum_{j=0}^{l}(-1)^{j}\binom{l}{j} \frac{(l+m-j-1)!}{(m-j)!} \stackrel{\mathrm{tbs}}{=} 0,
\end{equation}
which is true by Lemma~\ref{lemma:binomial_sum}. In order to show that the lhs of Equation~\eqref{eq:binomial_sum_induction} is also zero 
for $\mathrm{min}(l,m)=m$, Equation~\eqref{eq:binomial_sum_jsplit} is reformulated 
\begin{equation}
 \sum_{j=1}^{\mathrm{min}(l,m)}(-1)^{j+1}\binom{m}{j} \frac{(l+m-j-1)!}{(l-j)!}\stackrel{\mathrm{tbs}}{=}\frac{(l+m-1)!}{l!}
\end{equation}
The term on the rhs is again the `missing' summand for $j=0$ leading to
\begin{equation}
 \sum_{j=0}^{m}(-1)^{j}\binom{m}{j} \frac{(l+m-j-1)!}{(l-j)!} \stackrel{\mathrm{tbs}}{=} 0,
\end{equation}
which is again true by Lemma~\ref{lemma:binomial_sum} for $m\leq l$.
\end{proof}
%------------------ Lemma 3
It remains to prove the following combinatoric identity, which we used in the proof of Lemma \ref{lemma:induction}.
\begin{lemma}
For all $l,m\in\mathbb{N}$, $l\leq m$ it holds that
 \begin{equation}
 0 = \sum_{j=0}^{l}(-1)^{j}\binom{l}{j} \frac{(l+m-j-1)!}{(m-j)!}.
 \label{eq:binomial_sum}
\end{equation}
\label{lemma:binomial_sum}
\end{lemma}
\begin{proof}
 For all $l\in\mathbb{N}$ and $x\in\mathbb{R}$ we can employ the binomial formula
 \begin{equation}
  \left(1+\frac{1}{x}\right)^{l} = \sum_{j=0}^l\binom{l}{j}\frac{1}{x^j}.
 \label{eq:binomial_theorem}
 \end{equation}
Multiplication with $x^{l+m-1}$ on both sides yields
 \begin{equation}
  x^{l+m-1}\left(1+\frac{1}{x}\right)^{l} = \sum_{j=0}^l\binom{l}{j}x^{l+m-j-1}.
  \label{eq:binomal_first_step}
 \end{equation}
The procedure is as follows: we take the $(l-1)$-th derivative with respect to $x$ on both sides and then set $x=-1$. The lhs of 
Equation~\eqref{eq:binomal_first_step} is in the following denoted as
\begin{equation}
 \RN{3}(x):=  x^{l+m-1}\left(1+\frac{1}{x}\right)^{l}
 \label{eq:RN3_x}
\end{equation}
and the rhs is 
\begin{equation}
 \RN{4}(x):= \sum_{j=0}^l\binom{l}{j}x^{l+m-j-1}.
 \label{eq:RN4_x}
\end{equation}
Applying the Leibniz rule of differentiation to $\RN{3}$ yields
\begin{equation}
\begin{split}
\left(\frac{d}{dx}\right)^{l-1}\RN{3}(x)=\sum_{j=0}^{l-1}\binom{l-1}{j}&\left[\left(\frac{d}{dx}\right)^{l-1-j}x^{l+m-1}\right]\\[0.2em]\times&\left[\left(\frac
{ d } { dx } \right)^ { j
}\left(1+\frac{1}{x}\right)^{l}\right]
\end{split}
\end{equation}
Each of the terms in this sum contains a factor $(1+1/x)^p$ where $p\geq1, p\in\mathbb{N}$ since we take no more than $l-1$ derivatives. Setting $x=-1$, the 
factor $(1+1/x)^p$ becomes zero, i.e.  
\begin{equation}
 \left(\frac{d}{dx}\right)^{l-1}\RN{3}(-1)=0.
 \label{eq:lhs_after_derv_x-1}
\end{equation}
Taking the $(l-1)$-th derivative of $\RN{4}$ yields
 \begin{equation}
   \left(\frac{d}{dx}\right)^{l-1}\RN{4}(x) = \sum_{j=0}^l\binom{l}{j} \frac{(l+m-j-1)!}{(m-j)!} x^{m-j}.
 \end{equation}
Notice that $m-j\geq 0$ since $m\geq l$ and $j\leq l$. By inserting $x=-1$, we get
\begin{equation}
  \left(\frac{d}{dx}\right)^{l-1}\RN{4}(-1) =  \sum_{j=0}^l\binom{l}{j} \frac{(l+m-j-1)!}{(m-j)!} (-1)^{m-j}.
 \label{eq:rhs_after_derv_x-1}
\end{equation}
Putting the lhs, Equation~\eqref{eq:lhs_after_derv_x-1}, and the rhs, Equation~\eqref{eq:rhs_after_derv_x-1}, together and dividing both sides by $(-1)^m$ 
yields Equation~\eqref{eq:binomial_sum}.
\end{proof}

%%%%%%%%%%%%%%%%%%%%%%%%%%%%%%%%%%%%%%%%%%%%%%%% Proof of general formula for (0a| ra_2m|0b)
\section{Proof of general formula for \boldmath{$(0_a|r_a^{2m}|0_b)$}}
\label{app:proof_sra2m}
In this appendix, we prove that Equation~\eqref{eq:sra2ms} is valid for all $ m\in\mathbb{N}$.
%------------------------------------------- Theorem 4
\begin{theorem}
Equation~\eqref{eq:sra2ms},
\begin{align*}
 (0_a|r_a^{2m}|0_b) = \frac{\pi^{3/2}\exp(-\rho R_{ab}^2)}{2^mc^{m+3/2}}\sum_{j=0}^mI_j^{\alpha,\beta,m}(R_{ab}^2)  
\end{align*}
 is valid for all $ m\in\mathbb{N}$.
\label{theorem:sra2msint}
\end{theorem}

\begin{proof}
The matrix element~$(0_a|r_a^{2m}|0_b)$ as given in Equation~\eqref{eq:sra2ms_def} can be rewritten as
\begin{equation}
    (0_a|r_a^{2m}|0_b) = \exp(-\rho R_{ab}^2)\int \exp(-c r_p^2) r_a^{2m}d\mathbf{r},\label{eq:eqinlemmaV}
\end{equation}
where $\rho =\alpha\beta/c$, $c = \alpha + \beta$, $\mathbf{r}_p=\mathbf{r}-\mathbf{R}_p$ and 
\begin{equation}
  \mathbf{R}_p = \frac{\alpha \mathbf{R}_a + \beta \mathbf{R}_b}{c}\,.
\end{equation}
 This is clear by inserting Equation~\eqref{eq:SHG_primitive} and~$Y_{0,0}(\theta,\phi)=\frac{1}{\sqrt{4\pi}}$ into Equation~\eqref{eq:sra2ms_def} and applying 
the Gaussian product rule
\begin{equation}
 \exp(-\alpha r_a^2)\exp(-\beta r_b^2)=\exp(-\rho R_{ab}^2)\exp(-c r_p^2).
\end{equation}
Now we define the integral over a primitive $s$ function at center $\mathbf{R}_p$ multiplied with the operator $r_a^{2m}$ as
 \begin{equation}
  \RN{5}_m := \int \exp(-c r_p^2) r_a^{2m}d\mathbf{r},
 \end{equation}
 where $m\in\mathbb{N}$. Note that we have dropped the indices writing $\RN{5}_m$ instead of $\RN{5}_m^{\alpha,\beta,\mathbf{r}_a,\mathbf{r}_b}$. In the 
remainder of this proof, we explicitly calculate this Gaussian integral.

 \medskip
 
We start by rewriting the operator $r_a^{2m}$ in expression $\RN{5}_m$ in terms of $\mathbf{R}_{pa} = \mathbf{R}_p-\mathbf{R}_a$,
\begin{align}
   \RN{5}_m  & = \int \exp(-c r_p^2) \left|\mathbf{r}-\mathbf{R}_p+\mathbf{R}_p-\mathbf{R}_a\right|^{2m}d\mathbf{r} \\
             & = \int \exp(-c r_p^2) \left[r_p^2+2\mathbf{r}_p\cdot\mathbf{R}_{pa}+R_{pa}^2\right]^{m}d\mathbf{r}_p.
\end{align}
where $R_{pa}=|\mathbf{R}_{pa}|$. Employing a trinomial expansion yields
\begin{equation}
  \begin{split}
       \RN{5}_m  =  \int & \exp(-c r_p^2) \sum_{\substack{i+j+k=m\\i,j,k\in\mathbb{N}}} \binom{m}{i,j,k} 
r_p^{2i}\\&\times2^j(\mathbf{r}_p\cdot\mathbf{R}_{pa})^jR_{pa}^{2k}d\mathbf{r}_p,
  \end{split} 
  \label{eq:after_triexp}
 \end{equation}
where the multinomial coefficient is defined as
 \begin{align}
\binom{m}{i,j,k} := \frac{m!}{i!j!k!}\,.
 \label{eq:trinomial_expansion_coeff}
 \end{align}
Introducing the unit vector $\mathbf{\hat{R}}_{pa}$ in direction of $\mathbf{R}_{pa}$ yields
\begin{equation}
  \begin{split}
       \RN{5}_m  = &\sum_{i+j+k=m} 2^j\binom{m}{i,j,k}R_{pa}^{2k} |\mathbf{R}_{pa}|^j \\&\times\int \exp(-c r_p^2) 
r_p^{2i}(\mathbf{r}_p\cdot\mathbf{\hat{R}}_{pa})^jd\mathbf{r}_p.
   \end{split} 
 \end{equation}
Because of rotational symmetry, the integral can not depend on the 
direction of $\mathbf{R}_{pa}$. So without loss of generality we can take $\mathbf{\hat{R}}_{pa} = \mathbf{e}_z$, where $\mathbf{e}_z$ is the unit vector in 
$z$ direction.
In order to remove parameter $c$ from the integral, we substitute $\mathbf{r}_c := \sqrt{c}\mathbf{r}_p$,  
\begin{equation}
 \begin{split}
 \RN{5}_m = & \sum_{i+j+k=m} 2^j\binom{m}{i,j,k} R_{pa}^{2k+j} c^{-\frac{3}{2}-i-\frac{j}{2}}\\&\quad\times 
\int \exp(-r_c^2) 
r_c^{2i}(\mathbf{r}_c \cdot \mathbf{e}_z)^jd\mathbf{r}_c.
 \end{split}
 \label{eq:rn3_subs}
\end{equation}
$\RN{5}_m$ is non-zero only for even $j$ (since for odd $j$ the integrand is odd with respect to the reflection of $\mathbf{r}_c$ onto 
$-\mathbf{r}_c$) and so we can rewrite Equation~\eqref{eq:rn3_subs} as follows,
\begin{equation}
 \begin{split}
 \RN{5}_m = &\sum_{i+2j+k=m} 2^{2j}\binom{m}{i,2j,k} R_{pa}^{2k+2j} c^{-\frac{3}{2}-i-j}\\&\quad\times  \int \exp(-r_c^2) 
r_c^{2i}(\mathbf{r}_c \cdot \mathbf{e}_z)^{2j}d\mathbf{r}_c.
 \end{split}
\end{equation}
We introduce spherical coordinates with $\theta$ being the angle between $\mathbf{r}_c$ and the $z$-axis, i.e. $\mathbf{r}_c\cdot \mathbf{e}_z = r_c 
\cos\theta$,
\begin{equation}
 \begin{split}
 \RN{5}_m &=  \sum_{i+2j+k=m} 2^{2j}\binom{m}{i,2j,k} R_{pa}^{2k+2j} c^{-\frac{3}{2}-i-j}  \int_0^{2\pi}d\phi\\& \quad\times \int_0^{\pi}{\sin\theta 
(\cos\theta)^{2j}d\theta} \int_0^{\infty} r_c^2 \exp(-r_c^2)r_c^{2i}r_c^{2j}dr_c\,.
 \end{split}
\end{equation}
The integrals over $\theta$, $\phi$ and $r_c$ are evaluated explicitly. The integral over $\theta$ is obtained by substitution and the 
integral over $r_c$ is tabulated, for example, in Ref~\onlinecite{Bronstein_book}.
\begin{equation}
 \begin{split}
 \RN{5}_m  =  & \sum_{i+2j+k=m} 2^j\binom{m}{i,2j,k} R_{pa}^{2k+2j} c^{-\frac{3}{2}-i-j}\\&\times \frac{\pi^{3/2}}{1+2j}\frac{(1+2i+2j)!!}{2^{i}}\,.
\end{split}
\end{equation}
Employing that $\mathbf{R}_{pa}= \beta(\mathbf{R}_b-\mathbf{R}_a)/c$ yields
\begin{equation}
 \begin{split}
 \RN{5}_m = & \sum_{i+2j+k=m} 2^{j}\binom{m}{i,2j,k}c^{-\frac{3}{2}-i-3j-2k} \frac{\pi^{3/2}}{1+2j}\\&\quad\times\frac{(1+2i+2j)!!}{2^{i}} \beta^{2k+2j} 
R_{ab}^{2k+2j}\,.
 \end{split}
\end{equation}
In order to sort the sum by powers of $R_{ab}^{2}$, we introduce the Kronecker delta,
\begin{equation}
 \begin{split}
 \RN{5}_m = {}& \sum_{l=0}^m\sum_{i+2j+k=m} \delta_{l,k+j}2^{j}\binom{m}{i,2j,k}c^{-\frac{3}{2}-m-l} 
\\[5pt]&\times\frac{\pi^{3/2}}{1+2j}\frac{(1+2i+2j)!!}{2^{i}} \beta^{2l}R_{ab}^{2l}\,, 
 \end{split}
\end{equation}
where we have also used that $m=i+2j+k$ and $k =l-j$ to manipulate the exponent of $c$. In the next step, the $k$-sum is eliminated by replacing $k$ by $l-j$
\begin{equation}
 \begin{split}
 \RN{5}_m = & {} \sum_{l=0}^m\sum_{\substack{i+j=m-l\\i,j\ge0;j\le l}} 2^{j}\binom{m}{i,2j,l-j}c^{-\frac{3}{2}-m-l}\\&\times 
\frac{\pi^{3/2}}{1+2j}\frac{(1+2i+2j)!!}{2^{i}}\beta^{2l} R_{ab}^{2l}\,.
 \end{split}
\end{equation}
Then the sum over $i=0,...,m$ is eliminated due to the
constraint $i=m-j-l$,
\begin{equation}\label{eq:equation16}
 \begin{split}
 \RN{5}_m = & {} \sum_{l=0}^m\sum_{j\ge0}^{\mathrm{min}(l,m-l)} 
2^{j}\binom{m}{m-j-l,2j,l-j}c^{-\frac{3}{2}-m-l}\\[5pt]&\times\frac{\pi^{3/2}}{1+2j}\frac{(1+2m-2l)!!}{2^{m-j-l}} \beta^{2l}R_{ab}^{2l}\,.
 \end{split}
\end{equation}
To complete the proof, we have to show that \eqref{eq:equation16} can be simplified as 
\begin{equation}
 \begin{split}
   \sum_{l=0}^m&\sum_{j\ge0}^{\mathrm{min}(l,m-l)} 2^{j}\binom{m}{m-j-l,2j,l-j}c^{-\frac{3}{2}-m-l}\\&\qquad\times 
\frac{\pi^{3/2}}{1+2j}\frac{(1+2m-2l)!!}{2^{m-j-l}} 
\beta^{2l} 
R_{ab}^{2l}\\[5pt]
       & \stackrel{tbs}{=}\frac{\pi^{3/2}}{2^mc^{m+3/2}} \sum_{l=0}^m 2^l\frac{(2m+1)!! }{(2l+1)!!} \binom{m}{l}\frac{\beta^{2l}}{c^l}R_{ab}^{2l}.
 \end{split}
\end{equation}
It is sufficient to show that each summand $l$ on the lhs is identical to the summand $l$ on the rhs, i.e. after some reduction of both sides we have
\begin{equation}
 \begin{split}
   &  (1+2m-2l)!! \sum_{j=0}^{\text{min}(l,m-l)}2^{2j}\binom{m}{m-j-l,2j,l-j}\frac{1}{1+2j} \\& \stackrel{tbs}{=}\frac{(2m+1)!! }{(2l+1)!!} \binom{m}{l}.
 \end{split}
 \label{eq:tbspiecewicesummands}
\end{equation}
This is easily shown employing Lemma~\ref{lemma:trinomial_sum} and
the identity $(2n+1)!! = \frac{(2n+1)!}{2^nn!}$.
\end{proof}
%-------------------------------------------- Lemma 5
The following identity was used for the proof of Theorem~\ref{theorem:sra2msint}.
\begin{lemma}
 It holds for all $m,l\in\mathbb{N}$ and $l\leq m$ that
 \begin{equation}\label{eq:eqVI}\begin{split}
  & \sum_{j=0}^{\text{min}(l,m-l)}2^{2j}\binom{m}{m-j-l,2j,l-j}\frac{1}{1+2j}\\&\quad = \frac{(2m+1)!}{(2l+1)!(1+2m-2l)!}.
 \end{split}\end{equation}
\label{lemma:trinomial_sum}
\end{lemma}
\begin{proof}
The hypergeometric function~$_2F_1$ is defined as 
\begin{align}
_2F_1(a,b;c;z) = \sum_{j=0}^{\infty} \frac{(a)_j(b)_j}{(c)_j}\frac{z^j}{j!}\,,\label{eq:def_hypergeo}
\end{align}
for $a,b,c,z\in\mathbb{R}, |z|<1$. Note that this series is also convergent for $z=1$, if  $c >0$ and $c > \mathrm{max}(a,b,(a+b))$. The notation $(q)_j$ 
in Equation~\eqref{eq:def_hypergeo} is the Pochhammer symbol which is defined for $j\in\mathbb{N}$ as
\begin{align}
(q)_j = \left\{ 
\begin{array}{lll}
1 && j=0 \\[0.5em]
q(q+1)\ldots (q+j-1)&& j\ge 1
\end{array}
\right.\,.
\end{align}
For negative integers~$q=-n,n\in\mathbb{N}$, the Pochhammer symbol simplifies to
\begin{align}
(-n)_j = 
\left\{ 
\begin{array}{lll}
1 && j=0 \\[0.5em]
(-1)^{j}n!/(n-j)!&& 1\le j\le n\\[0.5em]
0 && j \ge n+1
\end{array}
\right.\,.\label{eq:pochhammer_neg_int}
\end{align}
For positive real values $x\in \mathbb{R}_{>0}$, the Pochhammer symbol is given by
\begin{align}
(x)_j = \frac{\Gamma(x+j)}{\Gamma(x)}\,,\label{eq:pochhammer_pos}
\end{align}
where the Gamma function for~$t\in\mathbb{R}_{>0}$ is defined as
\begin{align}
\Gamma(t)= \int_0^\infty x^{t-1}e^{-x}\,dx\,.
\end{align}
For positive integers~$n\in\mathbb{N}_{>0}$, the Gamma function evaluates to
\begin{align}
\Gamma(n) = (n-1)!\,.\label{eq:gamma_integer}
\end{align}
Moreover, a duplication identity\cite{handbook_gamma_dup} holds for $t\in\mathbb{R}_{>0}$,
\begin{align}
\Gamma\left(t+{\tfrac{1}{2}}\right) = \frac{2^{1-2t}\sqrt{\pi}\,\Gamma(2t)}{\Gamma(t)}\,.\label{eq:gamma_dupl}
\end{align}
We denote the lhs of \eqref{eq:eqVI} by $\RN{6}_{m,l}$,
\begin{equation}
  \RN{6}_{m,l} := \sum_{j=0}^{\text{min}(l,m-l)}2^{2j}\binom{m}{m-j-l,2j,l-j}\frac{1}{1+2j},\label{eq:def_VI}
 \end{equation}
and rewrite Equation~\eqref{eq:def_VI} recalling $l\le m$:
\begin{align}
\RN{6}_{m,l}=
\binom{m}{l} \sum_{j=0}^{\mathrm{min}(l,m-l)}  \frac{\frac{l!}{(l-j)!}\frac{(m-l)!}{(m-l-j)!}}{2^{-2j}\frac{(2j+1)!}{j!}}\frac{1}{j!}.
\label{eq:VI_reformulated}
\end{align}
Rewriting Equation~\eqref{eq:VI_reformulated} yields
\begin{eqnarray}
\hspace{-4em}\RN{6}_{m,l}&\overset{\eqref{eq:pochhammer_neg_int}-\eqref{eq:gamma_dupl}}{=}&
\binom{m}{l} \sum_{j=0}^{\infty} \frac{(-l)_j(-(m-l))_j}{\left(\frac{3}{2}\right)_j}\frac{1}{j!}\label{eq:insert_pochhammer}
\\[0.5em]
&\overset{\eqref{eq:def_hypergeo}}{=}&
\binom{m}{l} \,_2F_1 \left(-l,-(m-l);\tfrac{3}{2};1\right)\,.\label{eq:intermed_proof_VI}
\end{eqnarray}
Since $(-l)_j=0$ for $j>l$ and $(-(m-l))_j=0$ for $j>m-l$, see Equation~\eqref{eq:pochhammer_neg_int}, we can replace the upper bound 
$\mathrm{min}(l,m-l)$ by $\infty$ in Equation~\eqref{eq:insert_pochhammer}.
We use  Gauss' hypergeometric theorem\cite{hypergeometric_series_book} with $a,b\in\mathbb{R}$, $c >0$ and $c > \mathrm{max}(a,b,(a+b))$,
\begin{align}
_2F_1(a,b;c;1) = \frac{\Gamma(c)\Gamma(c-a-b)}{\Gamma(c-a)\Gamma(c-b)}\,,\label{eq:Gauss_hypergeometric}
\end{align}
to evaluate $_2F_1 \left(-l,-(m-l);\frac{3}{2};1\right)$ from Equation~\eqref{eq:intermed_proof_VI}:
\begin{align}
_2F_1 &\left(-l,-(m-l);\tfrac{3}{2};1\right)\nonumber
\\[0.5em]
&\overset{\eqref{eq:Gauss_hypergeometric}}{=}
\frac{\Gamma\left(\frac{3}{2}\right)\Gamma\left(\frac{3}{2}+m\right)}{\Gamma\left(\frac{3}{2}+l\right)\Gamma\left(\frac{3}{2}+m-l\right)}
\\[0.5em]
&\overset{\eqref{eq:gamma_dupl}}{=}\frac{\displaystyle\frac{\Gamma(2)}{\Gamma(1)}\frac{\Gamma(2m+2)}{\Gamma(m+1)}}{\displaystyle\frac{\Gamma(2l+2)}{\Gamma(l+1)}
\frac{\Gamma(2m-2l+2)}{\Gamma(m-l+1)}}
\\[0.5em]
&\overset{\eqref{eq:gamma_integer}}{=}  \binom{m}{l}^{-1}\frac{(2m+1)!}{(2l+1)!(2m-2l+1)!}\,.\label{eq:eval_hypergeo}
\end{align}
By inserting Equation~\eqref{eq:eval_hypergeo} into Equation~\eqref{eq:intermed_proof_VI}, we obtain
\begin{align}
 \RN{6}_{m,l} = \frac{(2m+1)!}{(2l+1)!(1+2m-2l)!}\,.
\end{align}
\end{proof}

\bibliography{shg_integrals}

\clearpage


\section*{Supplementary material}

# Supporting information for:

# Fast evaluation of solid harmonic Gaussian integrals for local-resolution-of-the-identity methods and range-separated hybrid functionals

Dorothea Golze,[*,†] Niels Benedikter,[‡] Marcella Iannuzzi,[†] Jan Wilhelm,[†] and Jürg Hutter[†]

[†]*Department of Chemistry, University of Zürich, Winterthurerstrasse 190, CH-8057 Zürich, Switzerland, and* [‡]*QMath, Department of Mathematical Sciences, University of Copenhagen, Universitetsparken 5, 2100 København, Denmark*

E-mail: dorothea.golze@chem.uzh.ch

## 1   Integrals $(0_a|\mathcal{O}|0_b)$

The expressions for the *s*-type basic integrals and their scalar derivatives with respect to the square separation $R_{ab}^2$ are displayed in Table S1. We use the following abbreviations

$$\rho = \frac{\alpha\beta}{\alpha + \beta} \tag{1}$$

and

$$T = \rho R_{ab}^2. \tag{2}$$

The Boys function $F_n(x)$, which is defined as

$$F_n(x) = \int_0^1 \exp(-xt^2)t^{2n}dt, \tag{3}$$

**Table S1:** Expressions for the two-center integrals $(0_a|O|0_b)$ and their $n$-th derivative $(0_a|O|0_b)^{(n)}$ with respect to $R_{ab}^2$ for different operators. See Reference S3. $F_n(x)$ is the Boys function [Equation 3].

| operator | $(0_a\|O\|0_b)$ |
|---|---|
| $\delta(\mathbf{r})$ | $\left(\frac{\pi}{\alpha+\beta}\right)^{3/2} e^{-T}$ |
| $1/r$ | $\frac{2\pi^{5/2}}{\alpha\beta\sqrt{\alpha+\beta}} F_0(T)$ |
| $\mathrm{erf}(\omega r)/r$ | $\frac{2\pi^{5/2}}{\alpha\beta} \frac{\omega}{\sqrt{\alpha+\beta}\sqrt{\omega^2+\rho}} F_0\left(\frac{\omega^2 T}{\omega^2+\rho}\right)$ |
| $\mathrm{erfc}(\omega r)/r$ | $\frac{2\pi^{5/2}}{\alpha\beta\sqrt{\alpha+\beta}} \left[ F_0(T) - \frac{\omega}{\sqrt{\omega^2+\rho}} F_0\left(\frac{\omega^2 T}{\omega^2+\rho}\right)\right]$ |
| $\exp(-\omega r^2)/r$ | $\left(\frac{\pi}{\alpha+\beta}\right)^{3/2} \frac{2\pi}{\rho+\omega} \exp\left(-\frac{\omega T}{\rho+\omega}\right) F_0\left(\frac{T\rho}{\rho+\omega}\right)$ |
| $\exp(-\omega r^2)$ | $\left(\frac{\pi^2}{(\alpha+\beta)(\rho+\omega)}\right)^{3/2} \exp\left(-\frac{\omega T}{\rho+\omega}\right)$ |

| operator | $(0_a\|O\|0_b)^{(n)}$ |
|---|---|
| $\delta(\mathbf{r})$ | $(-\rho)^n \left(\frac{\pi}{\alpha+\beta}\right)^{3/2} e^{-T}$ |
| $1/r$ | $\frac{2\pi^{5/2}}{\alpha\beta\sqrt{\alpha+\beta}} (-\rho)^n F_n(T)$ |
| $\mathrm{erf}(\omega r)/r$ | $\frac{2\pi^{5/2}\omega}{\alpha\beta\sqrt{\alpha+\beta}\sqrt{\omega^2+\rho}} \left(-\frac{\omega^2\rho}{\omega^2+\rho}\right)^n F_n\left(\frac{\omega^2 T}{\omega^2+\rho}\right)$ |
| $\mathrm{erfc}(\omega r)/r$ | $\frac{2\pi^{5/2}}{\alpha\beta\sqrt{\alpha+\beta}} \left[ (-\rho)^n F_n(T) - \frac{\omega}{\sqrt{\omega^2+\rho}} \left(-\frac{\omega^2\rho}{\omega^2+\rho}\right)^n F_n\left(\frac{\omega^2 T}{\omega^2+\rho}\right)\right]$ |
| $\exp(-\omega r^2)/r$ | $\left(\frac{\pi}{\alpha+\beta}\right)^{3/2} \frac{2\pi}{\rho+\omega} \exp\left(-\frac{\omega T}{\rho+\omega}\right) \sum_{j=0}^{n} \binom{n}{j} \left(-\frac{\omega\rho}{\rho+\omega}\right)^{n-j} \left(-\frac{\rho^2}{\rho+\omega}\right)^j F_j\left(\frac{T\rho}{\rho+\omega}\right)$ |
| $\exp(-\omega r^2)$ | $\left(\frac{\pi^2}{(\alpha+\beta)(\rho+\omega)}\right)^{3/2} \left(-\frac{\omega\rho}{\rho+\omega}\right)^n \exp\left(-\frac{\omega T}{\rho+\omega}\right)$ |

is computed as explained in Reference S1. The derivative[S2] of the Boys function is given by

$$\frac{dF_n(x)}{dx} = -F_{n+1}(x). \tag{4}$$

# 2 Detailed integral timings

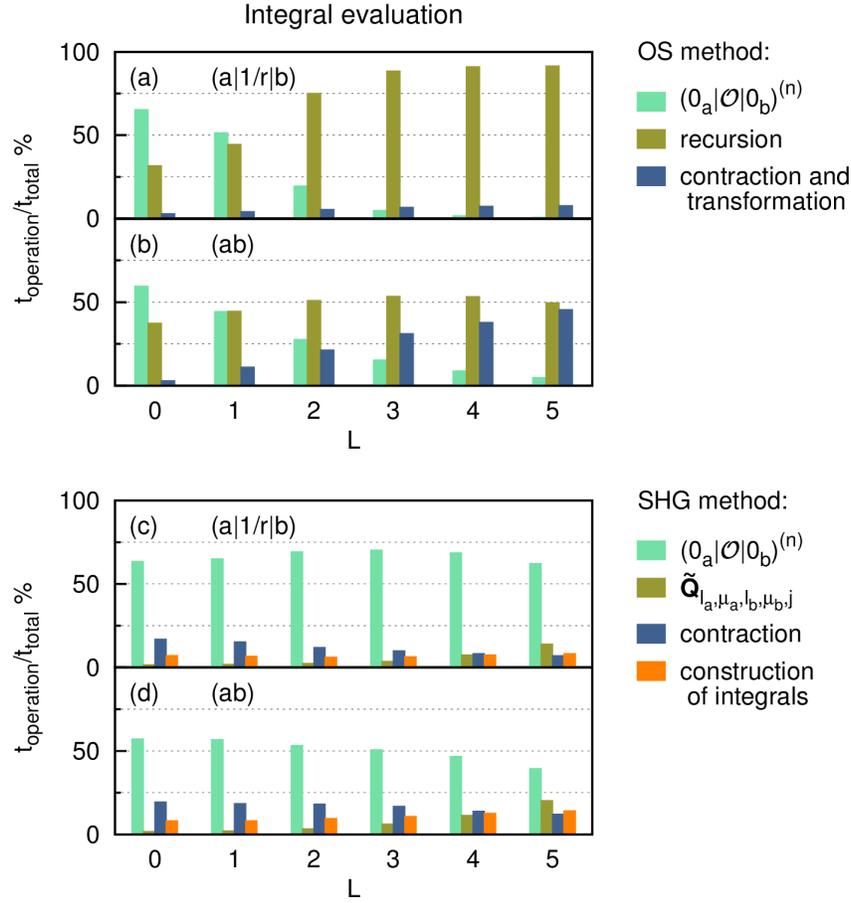

**Figure S1:** Time for single operations $t_{\text{operation}}$ relative to the total computational cost $t_{\text{total}}$ for evaluating spherical harmonic Gaussian integrals comparing the (a,b) OS and (c,d) SHG method. Timings are presented for the computation of the (a,c) Coulomb and (b,d) overlap integrals dependent on the $l$ quantum number at fixed contraction length $K = 7$. Employing the OS method, the fundamental integrals $(0_a|\mathcal{O}|0_b)$ and the auxiliary integrals $(-\rho)^{-n}(0_a|\mathcal{O}|0_b)^{(n)}$ have to be calculated, see Ref. S3 for details. Furthermore, the recursive procedure to obtain the primitive Cartesian integrals and the subsequent transformation and contraction step contribute significantly to the total computational cost. For the SHG method, the evaluation of $(0_a|\mathcal{O}|0_b)^{(n)}$, their contraction, the evaluation of the matrix elements $\widetilde{Q}^{c/s,c/s}_{l_a,\mu_a,l_b,\mu_b,j}$ and the construction of the integrals from the contracted monopole result $O^{(k)}_{l_a,l_b}$ and $\widetilde{Q}^{c/s,c/s}_{l_a,\mu_a,l_b,\mu_b,j}$ mainly contribute to the overall computational cost, see Figure 1 in the main text. Note that the absolute time for evaluating $(0_a|\mathcal{O}|0_b)^{(n)}$ is the same in both schemes and that the maximal derivative is $n_{\max} = l_{a,\max} + l_{b,\max}$.

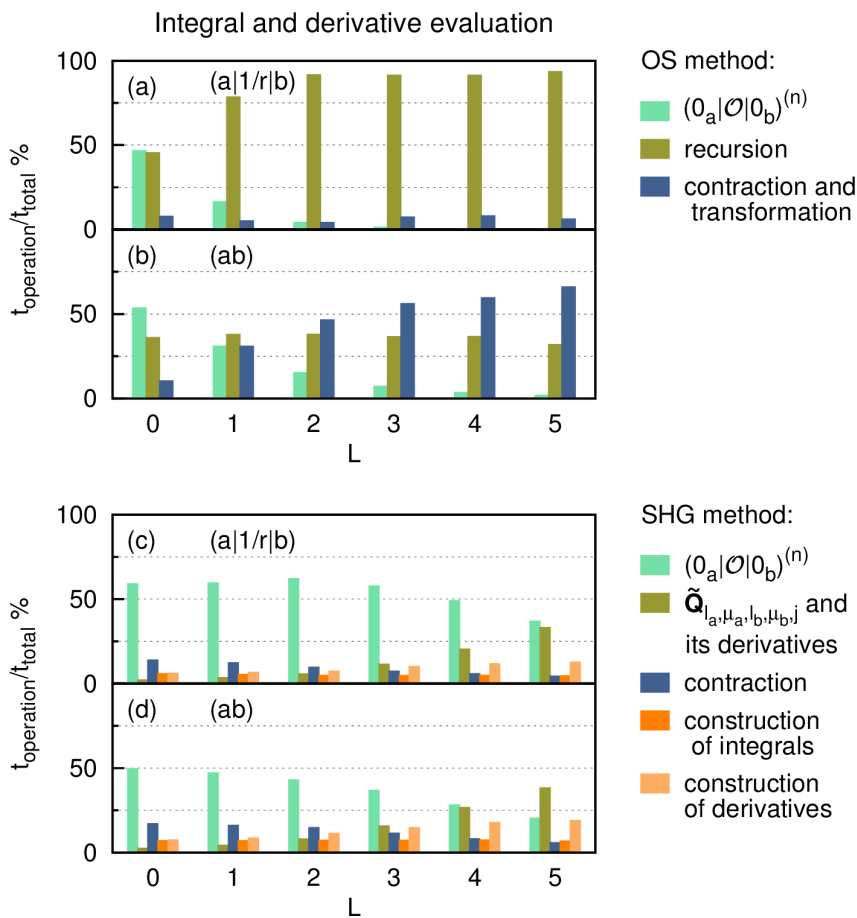

**Figure S2:** Time for single operations $t_{\text{operation}}$ relative to the total computational cost $t_{\text{total}}$ comparing the (a,b) OS and (c,d) SHG method. Timings are presented for the computation of the (a,c) Coulomb and (b,d) overlap integrals and their derivatives dependent on the $l$ quantum number at fixed contraction length $K = 7$. For the OS method, the construction of the derivatives of the primitive Cartesian integrals is part of the recursive procedure. For the SHG method, the derivatives of $\widetilde{Q}^{c/s,c/s}_{l_a,\mu_a,l_b,\mu_b,j}$ have to be additionally computed since they are required to construct the derivatives of the integrals as shown in Equation (44) in the main text.

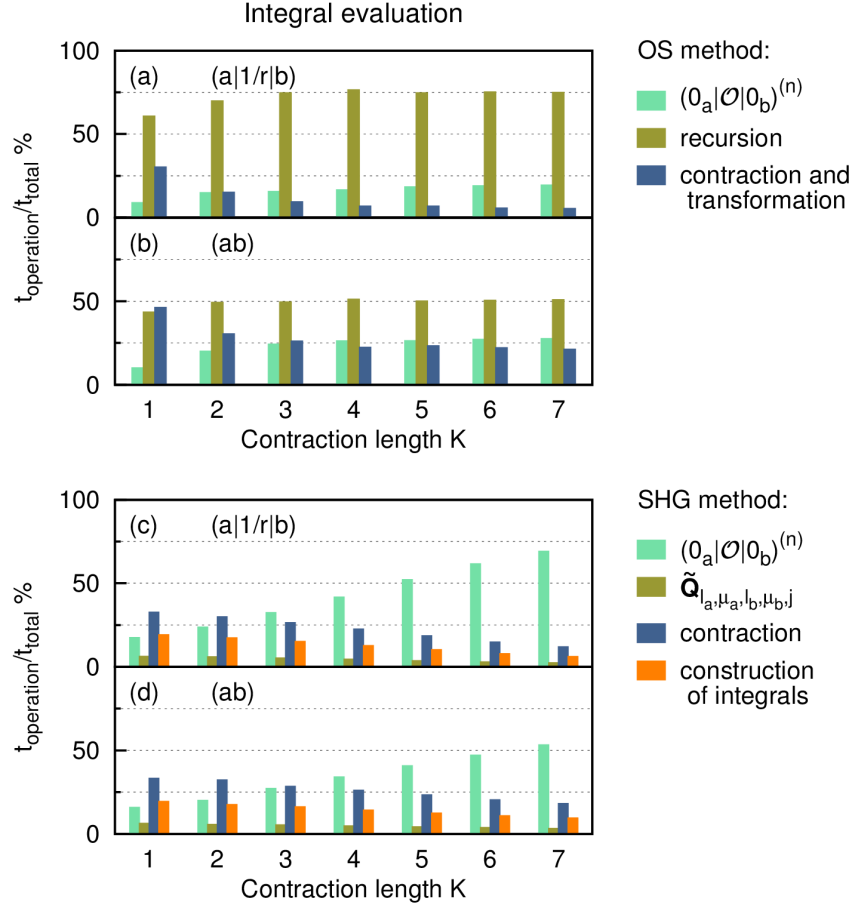

**Figure S3:** Time for single operations $t_{\text{operation}}$ relative to the total computational cost $t_{\text{total}}$ comparing the (a,b) OS and (c,d) SHG method. Timings are presented for the computation of the (a,c) Coulomb and (b,d) overlap integrals dependent on contraction length $K$ with fixed angular momentum $l = 2$. Note that for both, the SHG and the OS method, the timing measurements for the contraction step are ambiguous for $K = 1$ since the contraction routines have been optimized for contracted basis sets calling LAPACK routines for matrix-matrix multiplications. For $K = 1$, the time measured for this step is due to the initialization of the LAPACK routines and the corresponding arrays.

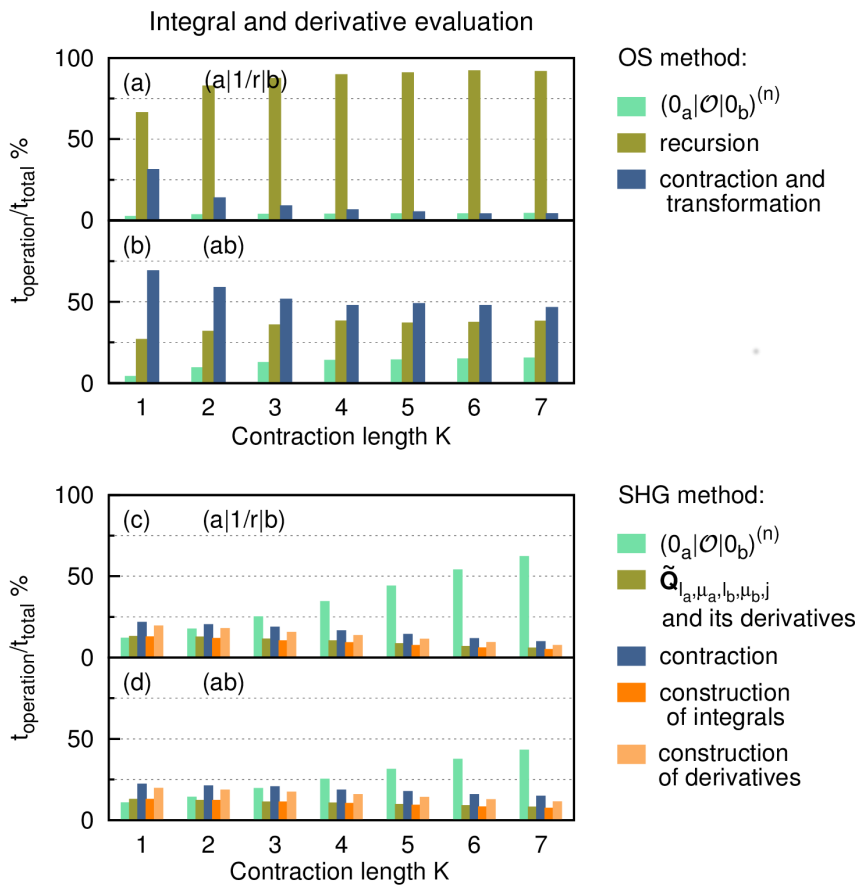

**Figure S4:** Time for single operations $t_{\text{operation}}$ relative to the total computational cost $t_{\text{total}}$ comparing the (a,b) OS and (c,d) SHG method. Timings are presented for the computation of the (a,c) Coulomb and (b,d) overlap integrals and their derivatives dependent on contraction length $K$ with fixed angular momentum $l = 2$. Note that for both, the SHG and the OS method, the timing measurements for the contraction step are ambiguous for $K = 1$ since the contraction routines have been optimized for contracted basis sets calling LAPACK routines for matrix-matrix multiplications. For $K = 1$, the time measured for this step is due to the initialization of the LAPACK routines and the corresponding arrays.

S6

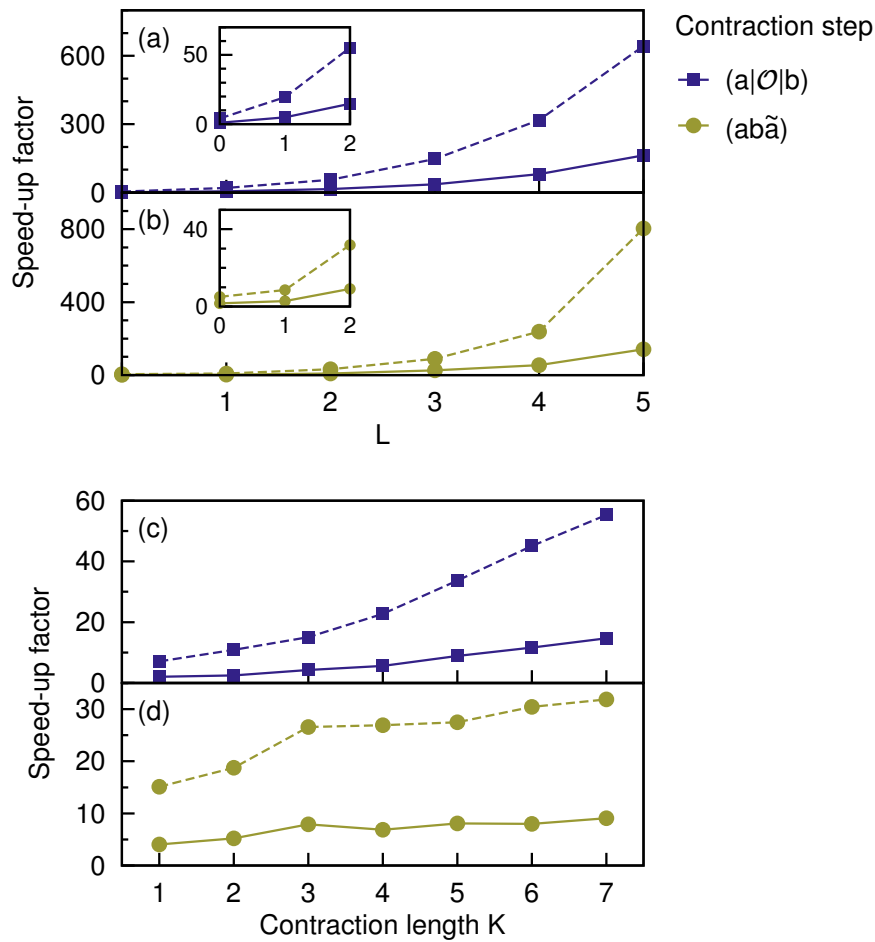

**Figure S5:** Speed-up of the contraction step dependent on the $l$ quantum number for the (a) two- and (b) three-index integrals at the fixed contraction length $K = 7$. Speed-up of the contraction step dependent on the contraction length $K$ for the (c) two- and (d) three-index integrals. The $l$ quantum number is fixed and set to $l = 2$. The speed-up factor is defined as the ratio OS/SHG. The solid line is the speed-up for the integrals and the dashed line is the speed-up for both, integrals + derivatives. Note that the contraction step for the two-index integrals $(a|O|b)$ is independent on the operator $O(\mathbf{r})$. Further note that for the OS method, the transformation to primitive spherical harmonic Gaussians and the contraction of the latter is done in one matrix-matrix multiplication step, i.e. we actually compare the contraction step of the SHG method to the contraction and transformation step of the OS method.

# 3 Basis sets

**Table S2:** `H-DZVP-MOLOPT-GTH`. Double-$\zeta$ valence plus polarization basis set for hydrogen.

|                  | Contraction coefficients | | |
|------------------|---------------|---------------|---------------|
| Exponents        | $s$           | $s$           | $p$           |
| 11.478000339908  | 0.024916243200  | -0.012512421400 | 0.024510918200 |
| 3.700758562763   | 0.079825490000  | -0.056449071100 | 0.058140794100 |
| 1.446884268432   | 0.128862675300  | 0.011242684700  | 0.444709498500 |
| 0.716814589696   | 0.379448894600  | -0.418587548300 | 0.646207973100 |
| 0.247918564176   | 0.324552432600  | 0.590363216700  | 0.803385018200 |
| 0.066918004004   | 0.037148121400  | 0.438703133000  | 0.892971208700 |
| 0.021708243634   | -0.001125195500 | -0.059693171300 | 0.120101316500 |

**Table S3:** `O-DZVP-MOLOPT-GTH`. Double-$\zeta$ valence plus polarization basis set for oxygen.

|                  | Contraction coefficients | | | | |
|------------------|------|------|------|------|------|
| Exponents        | $s$  | $s$  | $p$  | $p$  | $d$  |
| 12.015954705512  | -0.060190841200 | 0.065738617900  | 0.036543638800 | -0.034210557400 | 0.014807054400 |
| 5.108150287385   | -0.129597923300 | 0.110885902200  | 0.120927648700 | -0.120619770900 | 0.068186159300 |
| 2.048398039874   | 0.118175889400  | -0.053732406400 | 0.251093670300 | -0.213719464600 | 0.290576499200 |
| 0.832381575582   | 0.462964485000  | -0.572670666200 | 0.352639910300 | -0.473674858400 | 1.063344189500 |
| 0.352316246455   | 0.450353782600  | 0.186760006700  | 0.294708645200 | 0.484848376400  | 0.307656114200 |
| 0.142977330880   | 0.092715833600  | 0.387201458600  | 0.173039869300 | 0.717465919700  | 0.318346834400 |
| 0.046760918300   | -0.000255945800 | 0.003825849600  | 0.009726110600 | 0.032498979400  | -0.005771736600 |

**Table S4:** `O-TZV2PX-MOLOPT-GTH`. Triple-$\zeta$ valence plus double polarization basis set for oxygen.

|                  | Contraction coefficients | | | | | | | | |
|------------------|------|------|------|------|------|------|------|------|------|
| Exponents        | $s$ | $s$ | $s$ | $p$ | $p$ | $p$ | $d$ | $d$ | $f$ |
| 12.015954705512 | -0.060190841200 | 0.065738617900 | 0.041006765400 | 0.036543638800 | -0.034210557400 | -0.000592640200 | 0.014807054400 | -0.013843410500 | 0.002657486200 |
| 5.108150287385  | -0.129597923300 | 0.110885902200 | 0.080644802300 | 0.120927648700 | -0.120619770900 | 0.009852349400  | 0.068186159300 | 0.016850210400  | -0.007708463700 |
| 2.048398039874  | 0.118175889400  | -0.053732406400 | -0.067639801700 | 0.251093670300 | -0.213719464600 | 0.001286509800  | 0.290576499200 | -0.186696332600 | 0.378459897700 |
| 0.832381575582  | 0.462964485000  | -0.572670666200 | -0.435078312800 | 0.352639910300 | -0.473674858400 | -0.021872639500 | 1.063344189500 | 0.068001578700  | 0.819571172100 |
| 0.352316246455  | 0.450353782600  | 0.186760006700  | 0.722792798300 | 0.294708645200 | 0.484848376400  | 0.530504764700  | 0.307656114200 | 0.911407510000  | -0.075845376400 |
| 0.142977330880  | 0.092715833600  | 0.387201458600  | -0.521378340700 | 0.173039869300 | 0.717465919700  | -0.436184043700 | 0.318346834400 | -0.333128530600 | 0.386329438600 |
| 0.046760918300  | -0.000255945800 | 0.003825849600  | 0.175643142900 | 0.009726110600 | 0.032498979400  | 0.073329259500  | -0.005771736600 | -0.405788515900 | 0.035062554400 |

**Table S5:** `Cu-DZVP-MOLOPT-SR-GTH`. Double-$\zeta$ valence plus polarization short-range basis set for copper.

|                  | Contraction coefficients | | | | | | |
|------------------|------|------|------|------|------|------|------|
| Exponents        | $s$ | $s$ | $p$ | $p$ | $d$ | $d$ | $f$ |
| 5.804051150731 | 0.020918100390  | 0.045720931893  | -0.004381592772 | -0.021109803873 | 0.275442696345 | -0.101811263028 | -0.016523760157 |
| 2.947777593081 | -0.106208582202 | -0.094026024883 | 0.017185613995  | 0.020960301873  | 0.351705110927 | -0.207670618594 | 0.055142365254  |
| 1.271621207972 | 0.307397740339  | -0.110623536813 | -0.089805629814 | 0.233442747472  | 0.331635969640 | -0.224161904198 | -0.286656089760 |
| 0.517173767860 | 0.240805274553  | -0.742218346329 | 0.054415126660  | 0.369266430953  | 0.259386540456 | 0.176105738988  | -0.502349311598 |
| 0.198006620331 | -0.798718095004 | 2.208107372713  | 0.446326740476  | -1.405067129701 | 0.151105835782 | 0.210534119173  | -0.508940020577 |
| 0.061684232135 | -0.738671023869 | -1.720016262377 | 0.468516012555  | 1.042169799071  | 0.030634833418 | 0.902456275117  | 0.682110135764  |

**Table S6:** `H-LRI-MOLOPT-GTH`. Auxiliary basis set of hydrogen for the `MOLOPT` basis sets. Uncontracted basis set.

|                  | # of functions |   |   |   |
| Exponents        | s | p | d | f |
|------------------|---|---|---|---|
| 22.956000679816  | 1 | 0 | 0 | 0 |
| 11.437045132575  | 1 | 1 | 0 | 0 |
| 5.6981180297480  | 1 | 1 | 1 | 0 |
| 2.8388931498100  | 1 | 1 | 1 | 0 |
| 1.4143817790300  | 1 | 1 | 1 | 1 |
| 0.7046675275490  | 1 | 1 | 1 | 1 |
| 0.3510765846560  | 1 | 1 | 1 | 1 |
| 0.1749119456700  | 1 | 1 | 1 | 1 |
| 0.0871439169550  | 1 | 1 | 1 | 1 |
| 0.0434164872680  | 1 | 1 | 1 | 1 |

**Table S7:** `O-LRI-MOLOPT-GTH`. Auxiliary basis set of oxygen for the `MOLOPT` basis sets. Uncontracted basis set.

|                  | # of functions |   |   |   |   |
| Exponents        | s | p | d | f | g |
|------------------|---|---|---|---|---|
| 24.031909411024  | 1 | 0 | 0 | 0 | 0 |
| 16.167926705922  | 1 | 0 | 0 | 0 | 0 |
| 10.877281929506  | 1 | 1 | 0 | 0 | 0 |
| 7.3178994639200  | 1 | 1 | 1 | 0 | 0 |
| 4.9232568311730  | 1 | 1 | 1 | 1 | 0 |
| 3.3122151985280  | 1 | 1 | 1 | 1 | 0 |
| 2.2283561263540  | 1 | 1 | 1 | 1 | 1 |
| 1.4991692049680  | 1 | 1 | 1 | 1 | 1 |
| 1.0085947567100  | 1 | 1 | 1 | 1 | 1 |
| 0.6785514136050  | 1 | 1 | 1 | 1 | 1 |
| 0.4565084419100  | 1 | 1 | 1 | 1 | 1 |
| 0.3071247857670  | 1 | 1 | 1 | 1 | 1 |
| 0.2066240738900  | 1 | 1 | 1 | 1 | 1 |
| 0.1390102977340  | 1 | 1 | 1 | 1 | 1 |
| 0.0935218366000  | 1 | 1 | 1 | 1 | 1 |

**Table S8:** `Cu-LRI-MOLOPT-SR-GTH`. Auxiliary basis set of copper for the `MOLOPT` basis sets. Uncontracted basis set.

|                  | # of functions |   |   |   |   |   |   |
| Exponents        | s | p | d | f | g | h | i |
|------------------|---|---|---|---|---|---|---|
| 11.608102301462  | 1 | 0 | 0 | 0 | 0 | 0 | 0 |
| 8.3905970460850  | 1 | 0 | 0 | 0 | 0 | 0 | 0 |
| 6.0649119865960  | 1 | 1 | 0 | 0 | 0 | 0 | 0 |
| 4.3838545937940  | 1 | 1 | 1 | 0 | 0 | 0 | 0 |
| 3.1687485559560  | 1 | 1 | 1 | 1 | 0 | 0 | 0 |
| 2.2904426221370  | 1 | 1 | 1 | 1 | 1 | 0 | 0 |
| 1.6555833675850  | 1 | 1 | 1 | 1 | 1 | 1 | 0 |
| 1.1966928402980  | 1 | 1 | 1 | 1 | 1 | 1 | 1 |
| 0.8649964611020  | 1 | 1 | 1 | 1 | 1 | 1 | 1 |
| 0.6252388687580  | 1 | 1 | 1 | 1 | 1 | 1 | 1 |
| 0.4519366963750  | 1 | 1 | 1 | 1 | 1 | 1 | 1 |
| 0.3266699940390  | 1 | 1 | 1 | 1 | 1 | 1 | 1 |
| 0.2361244082660  | 1 | 1 | 1 | 1 | 1 | 1 | 1 |
| 0.1706760253360  | 1 | 1 | 1 | 1 | 1 | 1 | 1 |
| 0.1233684642700  | 1 | 1 | 1 | 1 | 1 | 1 | 1 |



\end{document}